\documentclass[11pt]{article}
\usepackage[utf8]{inputenc}
\usepackage{amsmath,amssymb,bbm}
\usepackage{authblk}
\usepackage{cprotect}

\usepackage{stmaryrd}
\usepackage{colortbl,xcolor}

\usepackage{subcaption}

\usepackage{fvextra}

\newcommand{\argmin}{\mathop{\arg\min}}
\newcommand{\bs}{\boldsymbol}
\newcommand{\diff}{\mathrm{d}}

\newcommand{\logit}{\mathrm{logit}}
\newcommand{\kdagger}{k^{\dagger}}

\newcommand{\ldagger}{\ell^{\dagger}}

\newcommand{\rdagger}{r^{\dagger}}

\newcommand{\setX}{\mathcal{X}}
\newcommand{\setU}{\mathcal{U}}
\newcommand{\xupper}{\eta}
\newcommand{\trunc}[1]{\left\llbracket #1 \right\rrbracket}
\newcommand{\sign}{\mathrm{sign}}
\newcommand{\round}[1]{[ #1 ]}
\newcommand{\cut}[1]{\mathfrak{C}(#1)}
\newcommand{\cox}{\mathrm{Cox}}
\newcommand{\satoh}{\mathrm{Satoh}}
\newcommand{\esssup}{\mathop{\text{esssup}}}

\newcommand{\pom}{\mathrm{POM}}
\newcommand{\npom}{\mathrm{NPOM}}
\newcommand{\cnpom}{\mathrm{N}^3\mathrm{POM}}
\newcommand{\dctm}{\mathrm{DCTM}}

\newcommand{\best}[1]{\textbf{\textcolor{blue}{#1}}}
\newcommand{\second}[1]{\textbf{\textcolor{red}{#1}}}

\usepackage{subcaption}

\usepackage{algorithm}
\usepackage{algorithmicx}
\usepackage{algpseudocode}
\usepackage{bbm}

\usepackage{stmaryrd}

\usepackage[whole]{bxcjkjatype}

\usepackage{graphicx}[dvipdfmx]

\usepackage{mdframed}
\usepackage{natbib}
\usepackage{hyperref}
\hypersetup{
    colorlinks = true,
	citecolor=blue,
	linkcolor=blue
}
\usepackage{booktabs}
\usepackage{multirow}

\usepackage{geometry}
\usepackage{enumerate}
\usepackage{amsthm}
\theoremstyle{definition}

\newtheorem*{theorem*}{Theorem}
\newtheorem{proposition}{Proposition}

\newtheorem{note}{Note}

\newtheorem{remark}{Remark}

\allowdisplaybreaks[4]

\usepackage{fourier}
\usepackage{comment}

\title{An interpretable neural network-based \\ non-proportional odds model for ordinal regression}
\author[1,2]{Akifumi Okuno}
\author[3,1]{Kazuharu Harada}
\affil[1]{The Institute of Statistical Mathematics}
\affil[2]{RIKEN Center for Advanced Intelligence Project}
\affil[3]{Tokyo Medical University}
\affil[ ]{\textit{okuno@ism.ac.jp}, \textit{haradak@tokyo-med.ac.jp}}
\date{\empty}

\begin{document}

\maketitle

\begin{abstract}
This study proposes an interpretable neural network-based non-proportional odds model (N$^3$POM) for ordinal regression. N$^3$POM is different from conventional approaches to ordinal regression with non-proportional models in several ways: (1) N$^3$POM is defined for both continuous and discrete responses, whereas standard methods typically treat the ordered continuous variables as if they are discrete, (2) instead of estimating response-dependent finite-dimensional coefficients of linear models from discrete responses as is done in conventional approaches, we train a non-linear neural network to serve as a coefficient function. 
Thanks to the neural network, N$^3$POM offers flexibility while preserving the interpretability of conventional ordinal regression. We establish a sufficient condition under which the predicted conditional cumulative probability locally satisfies the monotonicity constraint over a user-specified region in the covariate space. Additionally, we provide a monotonicity-preserving stochastic (MPS) algorithm for effectively training the neural network. We apply N$^3$POM to several real-world datasets.
\end{abstract}


\section{Introduction}

Ordinal regression modeling treats the response as measured on an ordinal scale and aims to understand the relationship between the response order and covariates~\citep{McCullagh1980-dl,Agresti2010-at}. In the context of ordinal regression modeling, response variables are typically assumed to be ordinal and discrete (e.g., the stage of cancer, the scores of wine quality). While standard regression-based approaches are mainly interested in the actual value of the response, this study focuses on thresholds of the responses; that is, the probability of the response being less than or equal to a specific threshold as a function of the covariates.

Let $d,J \in \mathbb{N}$ and consider $(G,X)$, a pair of such discrete ordinal response variables $G \in \{1,2,\ldots,J\}$ and their covariate $X \in \mathbb{R}^d$. 
A standard model for analyzing such a threshold is the proportional odds model~\citep[POM;][]{McCullagh1980-dl,mccullagh1989generalized}:
\[
    \logit(\mathbb{P}_{\pom}(G \le j \mid X=\bs x))
    =
    \alpha_j + \langle \bs \beta,\bs x \rangle \quad (j \in \{1,2,\ldots,J-1\}),
\]
where $\logit(z) = \log \frac{z}{1-z}$ is the logit function, 
and $\alpha_1,\alpha_2,\ldots,\alpha_{J-1} \in \mathbb{R},\bs \beta \in \mathbb{R}^d$ are parameters to be estimated. 
This model satisfies the proportionality assumption~\citep{McCullagh1980-dl}, which states that the regression coefficients are equal across all thresholds. However, the proportionality assumption is often considered to be violated~(see, e.g., \citet{long2006regression}). For instance, in the context of restaurant ratings, fundamental factors such as hygiene might be deemed more crucial in lower scores, whereas aspects like ingredient quality and wine selection could carry greater significance in higher scores. A practical example of the violation of the proportionality assumption is explored in \citet{williams2016understanding}.

One approach to circumvent the assumption violation is by leveraging the non-proportional odds model (NPOM, a.k.a., generalized ordinal logit model):
\[
    \logit(\mathbb{P}_{\npom}(G \le j \mid X=\bs x))
    =
    \alpha_j + \langle \bs \beta_j,\bs x \rangle \quad (j \in \{1,2,\ldots,J-1\}),
\]
which allows the coefficients to vary across the response thresholds. See, for example, \citet{mccullagh1989generalized}, \citet{Peterson1990-uu}, 
\citet{foresi1995conditional}, 
and \citet{williams2006generalized}. 
However, there remain several difficulties in leveraging NPOM while preserving interpretability. 

\begin{enumerate}[{(D-1)}]
\item
The first difficulty is the lack of monotonicity of the predicted conditional cumulative probability~(CCP); that is, the predicted CCP may violate monotonicity as $\hat{\mathbb{P}}_{\npom}(G \le j \mid X=\bs x) > \hat{\mathbb{P}}_{\npom}(G \le j+1 \mid X=\bs x)$ for some $(j,\bs x)$, because of the flexibility of the varying coefficients~\citep{tutz2022sparser,lu2022continuously}. 
Given that a simple POM with $\alpha_1 \le \alpha_2 \le \cdots \le \alpha_{J-1}$ is monotone in terms of the CCP, $\bs \beta_j=\bs \beta_*+\bs \delta_j$ with a prototype $\bs \beta_*$ is estimated with an $L^2$ penalty on $\bs \delta_j$ \citep{lu2022continuously}, and elastic net penalty \citep{Wurm2021-rv}. \citet{tutz2022sparser} further restricts the deviation as $\bs \delta_j=(j-J/2) \tilde{\bs \delta}$. However, specifying the proper penalty weights remains a problem.

\item
The second difficulty is the lack of the proximity guarantee of the estimated coefficients for adjacent thresholds. 
Given that it is more natural for the adjacent coefficients to be proximate, \citet{Tutz2016-gl} and \citet{ugba2021serp} incorporate a penalty between the coefficients of adjacent thresholds $\|\bs \beta_{j+1}-\bs \beta_{j}\|_2^2$. 
However, the adjacent penalty cannot be simply incorporated into the aforementioned monotonicity-guaranteed NPOM because of its incompatibility with the optimization algorithm.

\item
The third difficulty is the lack of extensibility to continuous responses. In some cases, we are interested in the probability that a continuous variable~(e.g., item prices, lifetime of people) is less than or equal to a threshold. 
Existing methods used to train the NPOM, which directly estimate countably many coefficients $\bs \beta_1,\bs \beta_2,\ldots,\bs \beta_J$, cannot be simply extended to the continuous response; a naive extension needs to estimate uncountably many coefficients $\{\bs \beta_u\}_{1 \le u \le J}$. 
\end{enumerate}

Our study distinguishes itself as the first to address all three of these difficulties simultaneously while maintaining a strong emphasis on interpretability. As a result, we propose a new approach called the \emph{neural network-based NPOM}~(N$^3$POM). Instead of directly modeling the cumulative probability, we train a neural network to estimate the coefficients of the linear model. This unique approach allows N$^3$POM to retain the interpretability that is characteristic of both conventional POM and NPOM. Additionally, we establish a sufficient condition for N$^3$POM to ensure the monotonicity constraint of the predicted CCP. To leverage this constraint effectively, we introduce a novel \emph{monotonicity-preserving stochastic} (MPS) algorithm. 
Finally, we demonstrate the effectiveness of N$^3$POM through experiments on a variety of synthetic and real-world datasets.

\section{Preliminaries and Related Works}

Section~\ref{subsec:problem_setting}--\ref{subsec:related_works} summarize problem-setting, symbols, and related works, respectively.

\subsection{Problem setting}
\label{subsec:problem_setting}
This section outlines the problem setting for this study. Let $n$, $d$, and $J$ be natural numbers, and let $H$ be a random variable representing an ordered response associated with a covariate $X \in \mathbb{R}^d$. 
We define the set 
\[
\setU := [1,J] = \{h \in \mathbb{R} \mid 1 \le h \le J\},
\]
where $H$ takes its values; the set $\setU$ is distinct from the discrete set $\{1,2,\ldots,J\}$, in which the responses of the conventional POM and NPOM take values.

Consider a dataset of $n$ observations, denoted as $\{(h_i,\bs x_i)\}_{i=1}^{n}$, where each is i.i.d. drawn from the pair of random variables $(H,X)$. The primary objective of this study is to explore the relationship between the covariate $X$ and the response $H$. To achieve this, we aim to predict the logit function applied to the conditional cumulative probability~(CCP), $\logit(\mathbb{P}(H \le u \mid X=\bs x))$ for $u \in \setU$, or equivalently, $\logit(\mathbb{P}(H > u \mid X=\bs x))$. Specifically, we seek to estimate the CCP, denoted as $\hat{\mathbb{P}}(H \le u \mid X=\bs x)$, in a manner that ensures its non-decreasing behavior with respect to $u \in \setU$, while holding $\bs x$ fixed. 
Throughout this study, we simply refer to this non-decreasing property as monotonicity.

\begin{note} 
In this study, instead of considering various intervals for the response range, such as $L,U$ where $-\infty < L < U < \infty$, we focus exclusively on the case where $L=1$ and $U=J \in \mathbb{N}$ (i.e., $\setU=[1,J]$). This specific case selection is primarily to straightforwardly compare our proposed N$^3$POM with the conventional POM and NPOM. In our numerical experiments, we utilize the conventional POM and NPOM as baseline models for the dataset $\{(h_i,\bs x_i)\}_{i=1}^{n}$, by discretizing the responses $h_i$ in advance. 
Although the set $\setU$ could potentially be expanded to include arbitrary intervals $[L,U]$, such a generalization is not deemed crucial for this study. To clarify this point, we here specifically consider the scenario where both the response $H^{\dagger}$ and its threshold $u^{\dagger}$ are within the interval $[L,U]$. Notably, a linear function $S(H^{\dagger}):=1 + (J-1)(H^{\dagger}-L)/(U-L)$ establishes a one-to-one monotonic relationship between the intervals $[L,U]$ and $[1,J]$. Consequently, we can calculate $\hat{\mathbb{P}}^{\dagger}(H^{\dagger} \le u^{\dagger} \mid X=\bs x)=\hat{\mathbb{P}}(S(H^{\dagger}) \le S(u^{\dagger}) \mid X=\bs x)$ for any $H^{\dagger},u^{\dagger} \in [L,U]$ right after obtaining the estimated CCP $\hat{\mathbb{P}}(H \le u \mid X=\bs x)$ for $H,u \in [1,J]$. 
\end{note}

\subsection{Symbols}
\label{subsec:symbols}

This study uses the following symbols. 
$\sigma(z)=1/(1+\exp(-z))$ represents the sigmoid function,  
$\logit(z)=\sigma^{-1}(z)=\log \frac{z}{1-z}$ represents the logit function, and 
$\rho:\mathbb{R} \to \mathbb{R}$ denotes a user-specified smooth activation function of a neural network (i.e., $\rho(z)=\tanh(z)$). 
$\langle \bs x,\bs x' \rangle = \bs x^{\top}\bs x'$ denotes an inner product of the vectors $\bs x,\bs x'$. $\|\bs x\|_{p}=\{x_1^p+x_2^p+\cdots+x_d^p\}^{1/p}$ for $\bs x=(x_1,x_2,\ldots,x_d), \, p \in \mathbb{N}$. 
$\emptyset$ denotes an empty set.

For clarity, $\nabla_{\bs \theta},\nabla_u$ are called by different names as the gradient (with respect to $\bs \theta$) and derivative (with respect to $u$), respectively. 
Particularly, for any function $f:\setU \to \mathbb{R}$, 
$f^{[1]}:\setU \to \mathbb{R}$ is called a weak derivative of $f$, if 
$f(u)=f(1)+\int_{1}^{u}f^{[1]}(\tilde{u}) \diff \tilde{u}$ holds for any $u \in \setU$. 
If $f$ is differentiable almost everywhere over $\setU$, $f^{[1]}$ is compatible with $\nabla_u f(u)$ except for all the (non-measurable) indifferentiable points in $f$; 
$f^{[1]}$ can be defined even if $f$ is indifferentiable at some points. 
Note that the weak derivative of $f$ is generally not unique, which is why we define a weak derivative for each function.

\subsection{Related works}
\label{subsec:related_works}

This section describes the related works. 
Among the various types of ordinal regression models, including the continuation-ratio logit and adjacent-categories logit models~\citep{Agresti2010-at}, this study places its primary emphasis on the cumulative logit model.

\subsubsection*{Transformation models}
Inspired by transformation models~\citep{box1964analysis}, \citet{foresi1995conditional} introduced a model for discrete conditional distributions, similar to NPOM, that is also termed distribution regression~\citep{chernozhukov2013inference}. 
\citet{Liu2017-zg} defines semiparametric linear transformation models. 
Seminal works explored most-likely transformation models (see, e.g., \citet{hothorn2018most}), defined as monotone transformations of $\langle \bs \beta(u),\bs c(\bs x)\rangle$, using general kernel basis expansion for $\bs \beta(u)$ and covariate transformation $\bs c(\bs x)$; therein, the coefficient function $\bs \beta(u)$ is defined for continuous response. 
Deep conditional transformation models \citep[DCTM;][]{baumann2021deep}, implemented as \verb|deeptrafo| package~\citep{kook2022estimating} in \verb|R| language, utilize deep neural networks for $\bs c$ and provide a sufficient condition for monotonicity when using Bernstein basis to compute $\bs \beta(u)$. However, unlike N$^3$POM, transformation models require predetermined kernels. 
Interpreting DCTM is more challenging than N$^3$POM as $\bs \beta(u)$ is a coefficient of the neural network output $\bs c(\bs x)$ (trained to guarantee the monotonicity adaptively to the dataset). 

\subsubsection*{Monotone neural network}
Another potentially possible model for estimating the CCP is a partially monotone neural network~\citep{daniels2010monotone}, which extends the univariate min-max network~\citep{sill1997monotonic} to multivariate settings. \citet{you2017deep} and \citet{liu2020certified} provide more flexible partially monotone deep neural networks. 
However, they lack interpretability.

\subsubsection*{Remaining extensions of ordinal regression} 
Ordinal regression has been extended to non-linear models; \citet{vargas2019deep} and \citet{vargas2020cumulative} replace the linear function $\langle \bs \beta,\bs x \rangle$ in POM with a neural network $m(\bs x)$, while N$^3$POM employs a neural network that outputs the coefficient vector $\bs \beta$. 
The simple intercept (SI) model in \citet{kook2022deep} first divides the covariate $\bs x=(\bs x',\bs x'')$ into $\bs x'$ of interest and the remaining $\bs x''$, and considers the combination of POM (with respect to $\bs x'$) and the neural network whose input is $\bs x''$. 
However, the aforementioned approaches consider prediction models independent of the threshold; they cannot capture the local relationship between the covariates and the response for each threshold. 
\citet{kook2022deep} also proposes a complex intercept~(CI) model, that can be regarded as a response-dependent prediction model (i.e., a non-linear extension of NPOM) with a discrete response. 
\citet{Thas2012-yj} proposes a probabilistic index model, which models the pairwise ordinal relationships of continuous and/or ordinal variables. Also, the proportional odds model in survival analysis \citep{bennett1983log,pettitt1984proportional} can be seen as a continuous extension of POM in ordinal regression. While they assume proportional regression coefficients, \citet{satoh2016logistic} further considers a non-proportional extension, which is also called the time-varying coefficient model. However, unfortunately, any sufficient condition to guarantee the monotonicity of their model is not provided.

\section{Proposed Model}
\label{section:OLCM}

In this section, we introduce a novel ordinal regression model designed specifically for continuous response variables. The proposed model, referred to as N$^3$POM, is elaborated upon in Section~\ref{subsec:proposed_model}. We also explore the monotonicity property of N$^3$POM in Section~\ref{subsec:monotonicity}. Subsequently, we present a parameter estimation algorithm in Section~\ref{subsec:parameter_estimation}

\subsection[Neural network-based non-proportional odds model (N3POM)]{Neural network-based non-proportional odds model~(N$^3$POM)}
\label{subsec:proposed_model}

To simultaneously address challenges outlined in the Introduction, including (D-1) the absence of monotonicity in the predicted CCP, (D-2) the absence of a proximity guarantee in the estimated coefficients for adjacent thresholds, and (D-3) the lack of extensibility to continuous responses, we introduce the \emph{neural network-based non-proportional odds model} (N$^3$POM): 
\begin{align}
    \logit(\mathbb{P}_{\cnpom}(H \le u \mid X=\bs x))
    &=
    \underbrace{
    a(u)
    +
    \langle \bs b(u), \bs x\rangle}_{=:f_u(\bs x)}, 
    \quad (u \in \setU=[1,J]). 
    \label{eq:cmccp}
\end{align}
Continuous functions $a:\setU \to \mathbb{R}$, $\bs b:\setU \to \mathbb{R}^d$, and their weak derivatives $a^{[1]},\bs b^{[1]}$, are defined later in \eqref{eq:au}--\eqref{eq:bu1}; we 
use the weak derivative 
$f_u^{[1]}(\bs x):=a^{[1]}(u) + \langle \bs b^{[1]}(u),\bs x \rangle$ to obtain the conditional probability density~(CPD) of $H \mid X$ as
\begin{align}
    q(u \mid X=\bs x)
    =
    \nabla_u \mathbb{P}_{\cnpom}(H \le u \mid X=\bs x)
    =
    \sigma^{[1]}(f_u(\bs x))
    f_u^{[1]}(\bs x).
\label{eq:conditional_pdf}
\end{align}
$\sigma(z)=1/(1+\exp(-z))$ denotes a sigmoid function and $\sigma^{[1]}$ is its derivative.

To obtain a non-negative CPD, the prediction model $f_u(\bs x)$ should be non-decreasing (with respect to $u \in \setU$ for any fixed $\bs x \in \setX$). Accordingly, we show a sufficient condition to guarantee the monotonicity in Section~\ref{subsec:monotonicity}. 
Using the CPD whose non-negativity is guaranteed, we provide a parameter estimation algorithm, which maximizes the log-likelihood defined later in \eqref{eq:proposed_log_lik}, in Section~\ref{subsec:parameter_estimation}. 
Note that estimating the CCP $\mathbb{P}_{\cnpom}(H \le u \mid X=\bs x)=\sigma(f_u(\bs x))$ is equivalent to estimating $\mathbb{P}_{\cnpom}(H >u \mid X=\bs x)
    =
    1-\mathbb{P}_{\cnpom}(H \le u \mid X=\bs x)
    =
    1-\sigma(f_u(\bs x))
    =
    \sigma(-f_u(\bs x))$.

\bigskip 
In the subsequent portion of this section, we will define the parametric functions $a$ and $\bs b$ and their weak derivatives $a^{[1]}$ and $\bs b^{[1]}$ to provide a comprehensive definition of N$^3$POM.

\begin{itemize}
\item 
Regarding the function $a:\setU \to \mathbb{R}$, this study considers a piece-wise linear functions with user-specified $R \in \mathbb{N}$ and $1 = j_1 < j_2 < \ldots < j_{R} = J$ (such that $\setU = [j_1,j_R]$):
\begin{align}
    a(u)
    :=   
    \begin{cases} 
    \alpha_r & (u = j_r, \: r \in \{1,2,\ldots,R\}) \\
    \alpha_{r-1} + s_{r-1}(u-j_{r-1})
    & (u \in (j_{r-1},j_{r}), \: r \in \{2,3,\ldots,R\}) \\
    \end{cases},
\label{eq:au}
\end{align}
where $s_{r-1}:=\frac{\alpha_{r}-\alpha_{r-1}}{j_{r}-j_{r-1}}$ denotes the slope of the line connecting two points $(j_{r-1},\alpha_{r-1})$ and $(j_r,\alpha_r)$. 
$-\infty < \alpha_1 \le \alpha_2 \le \cdots \le \alpha_R < \infty$ are parameters to be estimated. 
This piece-wise linear function is non-decreasing and satisfying $a(j_r)=\alpha_r$ ($r \in \{1,2,\ldots,R\}$). 
Consider a partition of the interval $\setU$: 
\[
    \setU_{r-1} := \begin{cases} 
        [j_{r-1},j_r) & (r \in \{2,3,\ldots,R-1\} \\
        [j_{R-1},j_R] & (r=R) \\
    \end{cases}
\]
satisfying $\bigcup_{r=2}^{R} \setU_{r-1}=\setU, \, \setU_{r} \cap \setU_{r'}=\emptyset \, (r \ne r')$. 
Then, we obtain the following weak derivative: 
\begin{align}
a^{[1]}(u)
:=
\sum_{r=2}^{R}
\mathbbm{1}(u \in \setU_{r-1})
s_{r-1}.
\label{eq:au1}
\end{align}

\item 
Regarding the vector-valued function $\bs b(u)=(b_1(u),b_2(u),\ldots,b_d(u)):\setU \to \mathbb{R}^d$, we employ independent neural networks~(NNs) $b_1,b_2,\ldots,b_d:\setU \to \mathbb{R}$ defined by
\begin{align}
    b_k(u)
    :=
    v_k^{(2)}
    +
    \sum_{\ell=1}^{L}
    w_{k,\ell}^{(2)}
    \rho(
    w_{k,\ell}^{(1)}u + v_{k,\ell}^{(1)}   
    ),
    \quad 
    (k \in \{1,2,\ldots,d \}),
\label{eq:bu}
\end{align}
where $\bs \psi_k:=\{w_{k,\ell}^{(2)},w_{k,\ell}^{(1)},v_k^{(2)},v_{k,\ell}^{(1)}\}_{\ell}$ is a set of weights to be estimated. 
As is widely acknowledged, the neural network $\bs b(u)$ possesses universal approximation capabilities, implying that $\bs b(u)$ can approximate any continuous function $\bs b_*(u)$ by increasing the number of hidden units $L \to \infty$~(refer to, for example, \citet{cybenko1989approximation}). 
Here, $\rho:\mathbb{R} \to \mathbb{R}$ denotes a smooth activation function. In general, we make the following assumptions:
\[
\text{(i)} \,  \rho \text{ is twice-differentiable,}
\quad \text{and} \quad
\text{(ii)} \, \rho^{[1]}_{\infty}:=\sup_{z \in \mathbb{R}}|\rho(z)'|<\infty.
\]
Common activation functions such as the sigmoid function $\rho(z)=1/(1+\exp(-z))$ and the hyperbolic tangent function $\rho(z)=\tanh(z):={\exp(z)-\exp(-z)}/{\exp(z)+\exp(-z)}$ satisfy the conditions (i) and (ii). It is important to note that if $w_{k,\ell}^{(2)}=0$ for all $\ell$, $b_k(u)$ reduces to a constant $v_k^{(2)}$, which is independent of $u$. Each entry of the derivative $\bs b^{[1]}(u)=(b_1^{[1]}(u),b_2^{[1]}(u),\ldots,b_d^{[1]}(u))$ of the NN $\bs b(u)$ is
\begin{align}
    b_k^{[1]}(u)
    :=
    \sum_{\ell=1}^{L}
    w_{k,\ell}^{(1)}
    w_{k,\ell}^{(2)}
    \rho^{[1]}(
    w_{k,\ell}^{(1)}u + v_{k,\ell}^{(1)}   
    ),
    \quad 
    (k \in \{1,2,\ldots,d\}).
\label{eq:bu1}
\end{align}
\end{itemize}

This study exclusively focuses on the use of the simple perceptron~\eqref{eq:bu} to obtain a straightforward sufficient condition for ensuring monotonicity, as described in Section~\ref{subsec:monotonicity}. While it is possible to employ a deep neural network instead, doing so would result in a more intricate sufficient condition.

\subsection[Monotonicity of N3POM]{Monotonicity of N$^3$POM} 
\label{subsec:monotonicity}

Given that the CCP $\mathbb{P}(H \le u \mid X=\bs x)$ should be non-decreasing with respect to $u \in \setU$ (for any fixed $\bs x \in \setX$), we consider the monotonicity of the N$^3$POM~\eqref{eq:cmccp}.
Equivalently, we focus on the monotonicity of the prediction model $f_u(\bs x)=a(u)+\langle \bs b(u),\bs x\rangle$.

\bigskip 
First,  we consider the function $a(u)$. 
Given that $f_u(\bs 0)=a(u)+\langle \bs b(u),\bs 0\rangle = a(u)$ should be monotone with respect to $u$, parameters in the function $a(u)$ should satisfy the inequality 
\begin{align}
\alpha_1 \le \alpha_2 \le \ldots \le \alpha_R;    
\label{eq:alpha_inequality}
\end{align}
we employ a re-parameterization 
\begin{align}
    \alpha_{r} 
    = 
    \alpha_{r} (\bs \varphi) 
    = 
    \begin{cases} 
    \phi & (r=1) \\
    \phi + \sum_{t=2}^{r} |\varphi_t| & (r=2,3,\ldots,R) \\
    \end{cases},
    \label{eq:re-parameterization}
\end{align}
equipped with newly-defined parameters $\bs \varphi = (\phi,\varphi_2,\varphi_3,\ldots,\varphi_R) \in \mathbb{R}^{R}$ to be estimated. 
By virtue of the aforementioned re-parameterization, the monotonicity inequality~\eqref{eq:alpha_inequality} always hods.

\bigskip 
Second, we consider the function $\bs b(u)$. 
We focus on the type of function $\bs b(u)$, which can satisfy the monotonicity constraint of the CCP. 
Unfortunately, we find that the function $\bs b(u)$ is limited to constant functions, if the N$^3$POM is monotone for all the covariates in the unbounded region, that is, the entire Euclidean space $\mathbb{R}^d$. See Proposition~\ref{prop:constant_f} with the proof shown in Supplement~\ref{supp:proofs}.  

\begin{proposition}
\label{prop:constant_f}
    Let $a:\setU \to \mathbb{R}$ be a function defined in \eqref{eq:au}, equipped with the re-parameterization \eqref{eq:re-parameterization}. 
    Let $\bs b:\setU \to \mathbb{R}^d$ be a continuous function. 
    If $f_u(\bs x):=a(u)+\langle \bs b(u),\bs x\rangle$ is non-decreasing with respect to $u \in \setU$ for any fixed $\bs x \in \mathbb{R}^d$, $\bs b(u)$ is a constant function. 
\end{proposition}

While N$^3$POM cannot be uniformly monotone over the entire covariate Euclidean space, covariates are usually expected to distribute in a specific bounded region. 
Therefore, instead of the unbounded Euclidean space $\mathbb{R}^d$, we consider a closed ball region for $\bs x$:
\begin{align*} 
\setX_{2}(\xupper):=\{\bs x \in \mathbb{R}^d \mid \|\bs x\|_{2} \le \xupper \}
\end{align*}
equipped with a user-specified parameter $\xupper>0$, and we 
prove that N$^3$POM can be monotone for all $\bs x \in \setX_{2}(\xupper)$. 
A sufficient condition provided for proving the monotonicity is satisfying an inequality
\begin{align}
    \min_{r=2,3,\ldots,R}s_{r-1}
    \, \ge \,
    \xupper \cdot \rho^{[1]}_{\infty} \cdot 
    \sqrt{
        \sum_{k=1}^{d}
        \left\{
            \sum_{\ell=1}^{L} |w_{k,\ell}^{(2)} w_{k,\ell}^{(1)}|
        \right\}^2
    },
    \label{eq:NN_weights_inequality}
\end{align}
with $\rho^{[1]}_{\infty}:=\sup_{z \in \mathbb{R}}|\rho^{[1]}(z)|$. 
For instance, $\rho_{\infty}^{[1]}=1$ for $\rho(z)=\tanh(z)$. 
$s_{r-1}:=\{a_{r}-a_{r-1}\}/\{j_{r}-j_{r-1}\}=\varphi_r^2/\{j_r-j_{r-1}\}$ represents a slope of the function $a(u)$. 
The following proposition holds, with the proof given in Supplement~\ref{supp:proofs}. 
See Figure~\ref{fig:illustration_of_proposition} for illustration.

\begin{proposition}
\label{prop:monotonicity}
Let $a:\setU \to \mathbb{R}$ be a function defined in \eqref{eq:au}, equipped with the re-parameterization \eqref{eq:re-parameterization}. 
Let $\bs b:\setU \to \mathbb{R}^d$ be a neural network defined in \eqref{eq:bu}. 
Suppose the inequality~\eqref{eq:NN_weights_inequality} holds. 
Then, $f_u(\bs x)$ is non-decreasing with respect to $u \in \setU$, for any fixed $\bs x \in \setX_{2}(\xupper)$. 
\end{proposition}

By specifying $\xupper > \max_i \|\bs x_i\|_{2}$, the estimated CCPs for all observed covariates $\{\bs x_i\}_{i=1}^{n}$ satisfy the monotonicity as $\{\bs x_i\}_{i=1}^{n} \subset \setX_2(\xupper)$. 
Note that the monotonicity may not be guaranteed (i.e., the log-likelihood to be optimized may not be well-defined) if $\xupper \le \max_i \|\bs x_i\|_2$. 
Therefore, in practice, we may specify $\eta := \max_i \|\bs x_i^{\text{train}}\|_2 + \varepsilon$ for training set and some $\varepsilon>0$. 
Although specifying $\varepsilon>0$ can be challenging depending on the problem setting, as long as we employ large enough $\varepsilon>0$ satisfying $\eta > \max_i \|\bs x_i^{\text{test}}\|_2$, estimated CCP is guaranteed to be monotone even for the test set.

\begin{figure}[!ht]
\centering 
\includegraphics[width=0.8\textwidth]{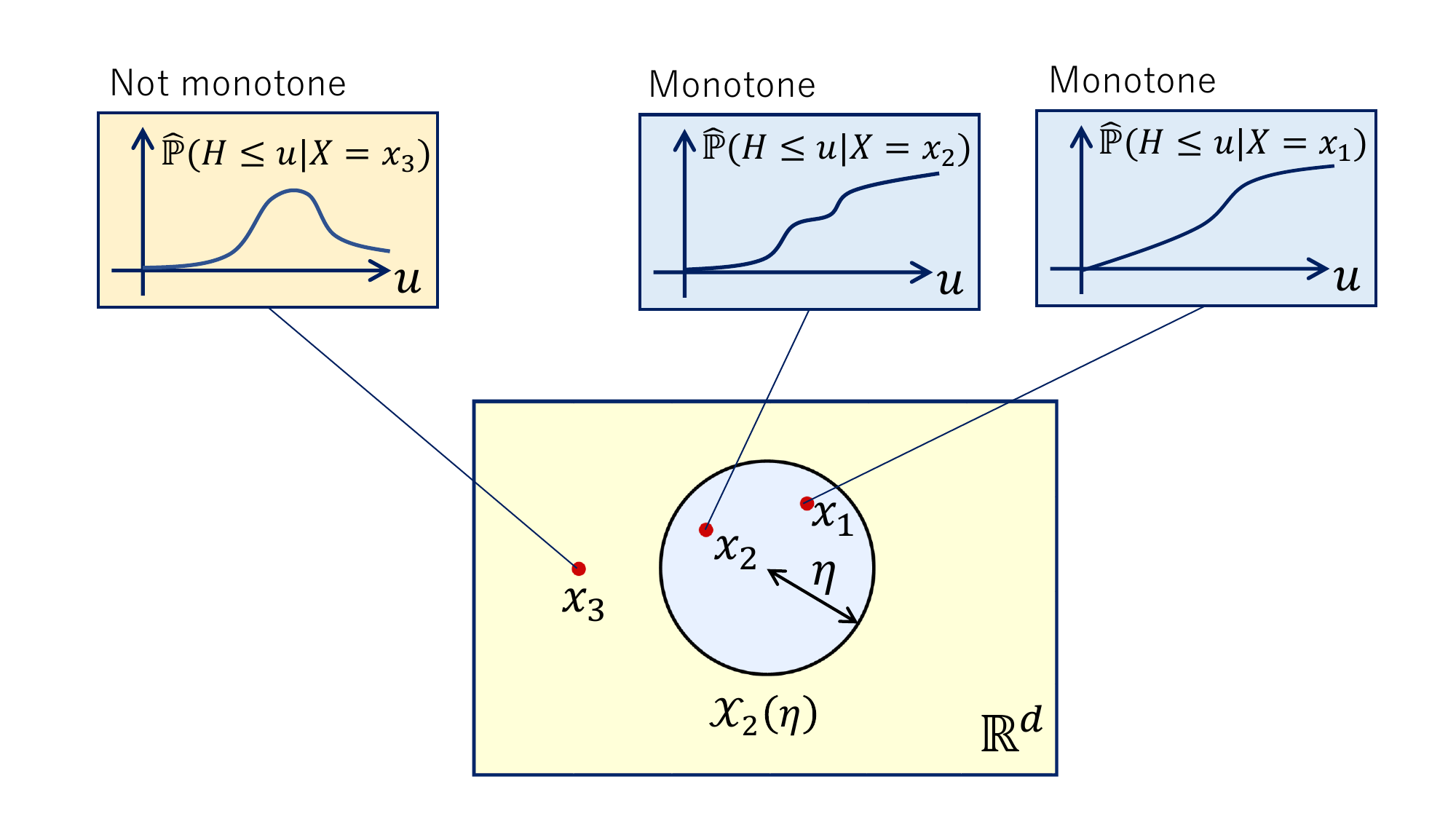}
\caption{Illustration of Proposition~\ref{prop:monotonicity}. 
The estimated CCP $\hat{\mathbb{P}}_{\cnpom}(H \le u \mid X=\bs x)=\sigma(\hat{f}_u(\bs x))$ is non-decreasing with respect to $u \in \setU$ (i.e., valid) if $\bs x \in \setX_2(\eta)$, while monotonicity is not guaranteed (i.e., invalid) if $\bs x \notin \setX_2(\eta)$. 
}
\label{fig:illustration_of_proposition}
\end{figure}

Here, we examine the relationship with Proposition~\ref{prop:constant_f}. When we set $\eta=\infty$ (which corresponds to $\setX_2(\xupper)=\mathbb{R}^d$) in inequality \eqref{eq:NN_weights_inequality}, we obtain the identity $\max_{k,\ell}|w_{k,\ell}^{(2)} w_{k,\ell}^{(1)}|=0$. This implies that either $w_{k,\ell}^{(2)}=0$ or $w_{k,\ell}^{(1)}=0$ for all $k,\ell$. Then, the function $b_k(u)$ becomes a constant function: $b_k(u) = \sum_{\ell=1}^{L} w_{k,\ell}^{(2)} \rho (v_{k,\ell}^{(1)}) + v_{k,\ell}^{(2)}$, as also mentioned in Proposition~\ref{prop:constant_f}. 
Conversely, selecting a smaller value for $\eta$ permits greater fluctuations in $\bs b(u)$ because the neural network weights ${w_{k,\ell}^{(2)},w_{k,\ell}^{(1)}}$ can have larger absolute values. Consequently, there exists a trade-off between model flexibility and the extent of the covariate region in which the CCP exhibits monotonic behavior.

The adverse property highlighted in Proposition~\ref{prop:constant_f} is not unique to our approach; it also applies to conventional NPOM methods. However, to the best of our knowledge, none of the existing studies have explicitly mentioned this proposition. Regarding the existing approaches, such as those discussed by \citet{tutz2022sparser}, \citet{lu2022continuously}, and \citet{Wurm2021-rv}, which penalize the distance between the coefficients and the prototype $\|\bs \beta_j-\bs \beta_*\|$, increasing the penalty weight corresponds to decreasing $\eta$ in our approach. Although they have made efforts to specify the penalty weights to ensure that the CCP maintains monotonicity over the observed covariates in the training dataset, they do not guarantee monotonicity in out-of-sample regions.

At last of this section, we turn our attention to a high-dimensional problem (i.e., $d \in \mathbb{N}$ is large). We assume that $\{\bs x_i\}_{i=1}^{n}$ are observed covariates i.i.d. drawn from a normal distribution $N(0,d^{-1}I)$ so that $\|\bs x_i\|_2=O_p(1)$ with $d \to \infty$, and we set $\eta:=\max_i \|x_i\|_2+\varepsilon$ for some small $\varepsilon>0$ (such that $\eta=O_p(1)$). In order to satisfy inequality \eqref{eq:NN_weights_inequality}, $|w_{k,\ell}^{(2)}w_{k,\ell}^{(1)}|$ should be of the order $d^{-1/2}$. Consequently, it is expected that $w_{k,\ell}^{(2)} \approx 0$ or $w_{k,\ell}^{(1)} \approx 0$ for all $k,\ell$, causing the coefficient function $b_k(u)$ to approximate a constant function: $b_k(u)\approx \sum_{\ell=1}^{L} w_{k,\ell}^{(2)}\rho(v_{k,\ell}^{(1)})+v_{k,\ell}^{(2)}$. 
Therefore, for high-dimensional problems, exploring more flexible coefficient functions that do not reduce to constant functions is a future research.

\subsection{Monotonicity-preserving stochastic~(MPS) algorithm}
\label{subsec:parameter_estimation}

Using the CPD $q(u \mid X=\bs x)$ defined in \eqref{eq:conditional_pdf}, we introduce a \emph{monotonicity-preserving stochastic}~(MPS) algorithm. The MPS algorithm is employed to estimate the parameters of the prediction model $f_u(\bs x)=a(u) + \langle \bs b(u),\bs x \rangle$ by optimally maximizing the weighted log-likelihood
\begin{align}
\ell_{\bs \zeta}(\bs \theta)
&:=
\sum_{i=1}^{n} \zeta_i \log q(h_i \mid X=\bs x_i)
\label{eq:proposed_log_lik} \\
&=
\sum_{i=1}^{n} \zeta_i \left\{
    \log \sigma^{[1]}(a(h_i) + \langle \bs b(h_i),\bs x_i \rangle))
    +
    \log (a^{[1]}(h_i) + \langle \bs b^{[1]}(h_i),\bs x_i \rangle))
\right\}, \nonumber 
\end{align}
under the sufficient condition constraint~\eqref{eq:NN_weights_inequality}. 
$\zeta_i \ge 0 \: (i \in \{1,2,\ldots,n\})$ are user-specified weights satisfying $\sum_{i=1}^{n} \zeta_i=1$ (e.g., $\zeta_1=\zeta_2=\cdots=\zeta_n=1/n$). 
For instance, we may employ $\zeta_i \propto \psi(\|\bs x_i-\bs \mu\|_2)$ with a robust mean $\bs \mu$ and some decreasing non-negative function $\psi$ for robust estimation as also discussed in \citet{croux2013robust}. 
\eqref{eq:proposed_log_lik} is a continuous variant of the log-likelihood for interval-censored data, which is typically used to train conventional NPOM. See Remark~\ref{remark:log-likelihood_of_NPOM}.

With an initial parameter $\bs \theta^{(0)}$, the MPS algorithm iteratively repeats the following two steps for $t=1,2,\ldots$ until convergence. 

\begin{enumerate}[{(i)}]

\item Compute a single step of the mini-batch gradient ascent algorithm~\citep{Goodfellow-et-al-2016} concerning the parameters $\bs \varphi=(\phi,\varphi_2,\ldots,\varphi_R)$ for $a(u)$ and $\bs \psi=\{w_{k,\ell}^{(2)},w_{k,\ell}^{(1)},v_k^{(2)},v_{k,\ell}^{(1)}\}_{k,\ell}$ for $\bs b(u)$. Explicit expressions for the gradients used in the gradient ascent can be found in Supplement~\ref{supp:gradient}.

\item With the coefficient
\[
    c = \min\left\{1, 
            \frac{
                \min_{r=2,3,\ldots,R}s_{r-1}
            }{
                    \xupper \cdot \rho^{[1]}_{\infty} \cdot 
    \sqrt{
        \sum_{k=1}^{d}
        \left\{
            \sum_{\ell=1}^{L} |w_{k,\ell}^{(2)} w_{k,\ell}^{(1)}|
        \right\}^2
    }
            }    
    \right\},
\]
we multiply $w_{k,\ell}^{(2)},w_{k,\ell}^{(1)}$ by $\sqrt{c}$ so as to satisfy the inequality \eqref{eq:NN_weights_inequality}, that is a sufficient condition to guarantee the monotonicity of the predicted CCP in a ball of radius $\eta$.
\end{enumerate}

Unlike the standard mini-batch gradient ascent algorithm, which consists only of step (i), step (ii) is essential to ensure the monotonicity of the predicted Conditional Cumulative Probability (CCP) $\hat{\mathbb{P}}_{\cnpom}(H \le u \mid X=\bs x)$ within a ball of radius $\eta$, 
as discussed in Section~\ref{subsec:monotonicity}. For adaptations to discrete responses, please refer to Supplement~\ref{supp:additive_perturbation}.

\begin{remark}[Relation to the likelihood for interval-censored data] 
\label{remark:log-likelihood_of_NPOM}
Define a function $\nu_i(\Delta)=\{\mathbb{P}(H \le h_i+\Delta \mid X=\bs x_i)-\mathbb{P}(H \le h_i \mid X=\bs x_i)\}/\Delta$ (satisfying $\lim_{\Delta \searrow 0}\nu_i(\Delta)=q(h_i \mid X=\bs x_i)$) and 
$j_1=1,j_2=2,\ldots,j_R=R=J$. 
Accordingly, the log-likelihood~\eqref{eq:proposed_log_lik} employed in this study corresponds to 
$\sum_{i=1}^{n} \zeta_i \log \{\lim_{\Delta \searrow 0}\nu_i(\Delta)\}$, while a widely-used log-likelihood for interval-censored data (i.e., $\{h_i\}$ taking values over the set $\{1,2,\ldots,J-1\}$) can be regarded as its discrete approximation $\sum_{i=1}^{n} \zeta_i \log \nu_i(1)$. See, e.g., \citet{simpson1996interval}. 
\end{remark}

\section{Experiments on Synthetic Datasets}
\label{sec:experiments_synthetic}

This section provides numerical experiments using synthetic datasets. 
Also see Supplement~\ref{supp:additional_experiments_synthetic} for additional experiments. 
\verb|R| source codes to reproduce the experimental results are provided in \url{https://github.com/oknakfm/N3POM}. 

\subsection{Synthetic datasets}
\label{subsec:synthetic_datasets}

In this experiment, we set $n=1000,d=2,J=7$. 
For the covariate $X$, we generate $r_1,r_2,\ldots,r_n \sim U([0,1])$ and 
$\theta_1,\theta_2,\ldots,\theta_n \sim U([0,2\pi))$ uniformly and randomly, and compute $\bs x_i=(x_{ij}) = (r_i \cos \theta_i, r_i \sin \theta_i) \in \mathbb{R}^2$. 
We consider the functions 
$a_*(u)=2u-9$ and 
\[
\bs b_*(u)=(b_{*1}(u),b_{*2}(u))=(-1+m_1 u^2 \, , \, 1+m_2 u^2)
\]
so that the continuous responses $h_1,h_2,\ldots,h_n$ are generated based on conditional distribution
$\mathbb{P}(H \le h_i \mid X = \bs x_i)
    =
    \sigma \left(
        a_*(h_i) + \langle \bs b_*(h_i),\bs x_i \rangle
    \right)$. 
$\{h_i\}$ are generated through inversion sampling in our implementation. 
To follow our setting, we further truncate $u$ to take values in the interval $\setU=[1,J]$ (using the function $\argmin_{\tilde{u} \in [1,J]}|u-\tilde{u}$), with $J=7$. The truncation is equivalent to $\max\{u, 1\},u,\min\{J,u\}$ if $u<1, 1\le u \le J, u>J$, respectively. 
Note that the truncation has only a small impact on parameter estimation in our experimental setting, as the randomly generated $u$ rarely falls outside the interval $[1,7]$ (i.e., truncation rarely occurs). 
For instance, the probability of truncation is approximately $0.0063$ in our experiments with $m_1=0.05,m_2=-0.05$, as discussed in Supplement~\ref{app:truncation_probability}. 
The observed covariates $\bs x_1,\bs x_2,\ldots,\bs x_n$ lie in the unit disk $\setX_2(1)=\{\bs x \in \mathbb{R}^2 \mid \|\bs x\|_2 \le 1\}$; the underlying function $f_u(\bs x)=a_*(u) + \langle \bs b_*(u),\bs x \rangle$ is non-decreasing with respect to $u \in \setU$ for all $\bs x \in \setX_2(1)$, as $f^{[1]}_u(\bs x)=\nabla_u f_u(\bs x)=2 - 2 m_1 u x_1 + 2 m_2 u x_2 \ge 0$. 
To compute the baselines, we discretize the observed continuous responses $h_i$ as
\begin{align}
\round{h_i}:=\argmin_{j \in \{1,2,\ldots,J\}} \|h_i-j\|_2;
\label{eq:discretization}
\end{align}
We train the POM and NPOM by leveraging $\round{h_i}$. 
Furthermore, we train the proposed N$^3$POM with $h_i,\round{h_i}$ and $\cut{\round{h_i}}$ for comparison, where $\cut{\cdot}$ is the random perturbation operator defined in Supplement~\ref{supp:additive_perturbation}.

\subsection{Experimental settings}
\paragraph{Model architecture:} We employ the proposed N$^3$POM~\eqref{eq:cmccp} defined in Section~\ref{subsec:proposed_model}. 
We specify $R=24$ for the function $a(u)$ and employ the regular intervals $1=j_1 < j_2 < j_3 < \cdots < j_{24}=7$; the reason behind choosing $R$ to be smaller than $n$ is to maintain the stability of the N$^3$POM optimization process. 
The number of hidden units in the neural network $\bs b(u)$ is $L=50$. 
The sigmoid activation function $\rho(z)=1/(1+\exp(-z))$ is also employed.

\paragraph{Initialization:} 
First, we compute the coefficient vectors $\hat{\bs \beta}_1,\hat{\bs \beta}_2,\ldots,\hat{\bs \beta}_{J-1}$ of the discrete NPOM by leveraging \verb|serp| package in \verb|R| language. 
Next, we initialize the neural network parameters so that the NN outputs approximate the \verb|serp| outputs over the discrete points (i.e., $b_k(j) \approx \hat{\beta}_{jk}$). 
See Supplement~\ref{supp:initialization} for details. 
The parameters of the function $a(u)$ are also parameterized by linear interpolation of  \verb|serp| outputs.

\paragraph{Optimization:} 
The weights in the log-likelihood~\eqref{eq:proposed_log_lik} are specified as 
$\zeta_i \propto 1/n_{r_i}^{0.5}$, where 
$n_r := |\{i : g_i \in \setU_r\}|$ and $g_{i} \in \setU_{r_i}$. 
To maximize the log-likelihood, we employ the MPS algorithm equipped with $\eta:=\max_{i=1,2,\ldots,n}\|\bs x_i\|_2+10^{-2}$. 
A mini-batch of size $16$ is uniformly and randomly selected from training samples (without replacement), and the number of iterations is $5000$. The learning rate is multiplied by $0.95$ for each $50$ iteration. 
Based on these settings, we apply the MPS algorithm to the following three types of observed responses: $h_i, \cut{\round{h_i}}, \round{h_i}$.

\paragraph{Baselines:} 
We utilize the following major implementations of ordinal regression in the \verb|R| language: 
(1) \verb|polr| function (logistic model) from the \verb|MASS| package, 
(2) \verb|ordinalNet| (\verb|oNet| function) from the ordinalNet package~\citep{Wurm2021-rv}, with options for non-proportional terms, cumulative logit, and hyperparameter $\alpha \in {0, 0.5, 1}$, 
(3) \verb|serp| function from the \verb|serp| package~\citep{ugba2021serp}, using a logit link and the "penalize" slope option. 
These implementations are employed for training the Proportional Odds Model (POM), Non-Proportional Odds Model (NPOM), and NPOM with both proportional and non-proportional terms (NPOM$^{\dagger}$). 
They are trained by maximizing the likelihood for interval-censored responses described in Remark~\ref{remark:log-likelihood_of_NPOM}. 
For \verb|ordinalNet|, the coefficients $\bs \beta_j$ are decomposed into $\bs \beta_*$, representing the proportional term, and ${\bs \delta_j}$, representing non-proportional terms. We compute NPOM with only the non-proportional term (referred to as NPOM) and NPOM with both proportional and non-proportional terms (NPOM$^{\dagger}$). 
Regarding penalty terms, \verb|ordinalNet| applies penalties as $\alpha\|\bs \beta\|_1 + (1-\alpha)\|\bs \beta\|_2^2 + \sum_{j=1}^{J-1}\{\alpha \|\bs \delta_j\|_1+(1-\alpha)\|\bs \delta_j\|_2^2\}$, where $\alpha$ values of 0, 0.5, and 1 correspond to ridge, elastic net, and lasso penalties, respectively. 
\verb|serp| penalizes adjacent coefficients as $\|\bs \beta_j - \bs \beta_{j-1}\|_2^2$. These baselines are applied to the discretized observed responses $\round{h_i}$.

\paragraph{Evaluation:} 
for each estimated coefficient $\hat{b}_k(u)$, we compute mean squared error~(MSE):
\[
    \text{MSE}(\hat{b}_k)
    :=
    \frac{1}{|\tilde{\setU}|}
    \sum_{\tilde{u} \in \tilde{\setU}}
    \{b_{*k}(\tilde{u}) - \hat{b}_{k}(\tilde{u})\}^2
    \quad (k=1,2),
\]
with $\tilde{\setU}:=\{1,1.05,1.1,1.15,1.2,\ldots,J\}$. 
For evaluating the POM and discrete NPOM, we employ linear interpolation for computing $\hat{b}_k(u)$ for non-integer $u$. 
We compute the MSE for each of the $20$ times experiments. 
However, it is possible for parameters to occasionally remain in severe local minima during neural network training, particularly when conducting multiple training runs. 
To mitigate this, we calculate the (robust) average of MSEs of the 20 experimental runs by removing the single highest and single lowest values in each experimental configuration. 
We also compute the standard deviation of MSE after removing the top/bottom 1 highest/lowest values for each setting.

\subsection{Results}

Experimental results for $\bs x_i=(r_i \cos \theta_i, r_i \sin \theta_i), \: r_i \sim U([0,1]), \theta_i \sim U([0,2\pi))$ are summarized in Table~\ref{table:synthetic_setting1}. 
For all the cases, the N$^3$POM trained using the observed continuous response $h_i$ shows the best scores. 
The N$^3$POM applied to the discretized observed response $\round{h_i}$ shows scores that are nearly equal to \verb|polr|. 
Even if the response is discretized, the score gets closer to that of the observed continuous response by leveraging the adaptation to the discrete responses shown in Supplement~\ref{supp:additive_perturbation}. 
The scores of the NPOM implemented by \verb|serp| package follow the scores of the N$^3$POM. 
\verb|ordinalNet| demonstrates the subsequent scores. 
Also see Supplement~\ref{supp:additional_experiments_synthetic} for additional experiments.

\begin{table}[!ht]
\centering
\caption{ 
Results of the MSE experiments on synthetic datasets, where the observed covariates are generated by the setting (i) $x_i=(r_i \cos \theta_i, r_i \sin \theta_i), r_i \sim U([0,1]),\theta_i \sim U([0,2\pi))$. 
Both coefficients $b_{*1}(u)=-1+m_1u^2,b_{*2}(u)=1+m_2u^2$ are response-dependent (i.e., not constant).
For the 20 experiments for each setting, the robust average and standard deviation (shown in parenthesis) for MSEs, are computed. The best score is \best{bolded and blue-colored}, while the second best score is \second{bolded and red-colored}.}
\label{table:synthetic_setting1}
\scalebox{0.85}{
\begin{tabular}{rrr|cc|cc}
\toprule 
\multirow{2}{*}{Model} & \multirow{2}{*}{Optimizer} & \multirow{2}{*}{Response} & \multicolumn{2}{c|}{$(m_1,m_2)=(0.05,-0.05)$} & \multicolumn{2}{c}{$(m_1,m_2)=(0.05,0.05)$} \\
 & & & MSE($\hat{b}_1$) & MSE($\hat{b}_2$) & MSE($\hat{b}_1$) & MSE($\hat{b}_2$) \\
\hline
\rowcolor{gray!30}
N$^3$POM & MPS & $h_i$ & \best{0.066} (0.155) & \best{0.122} (0.147)  & \best{0.052} (0.062) & \best{0.134} (0.138) \\
\rowcolor{gray!30}
N$^3$POM & MPS & $\cut{\round{h_i}}$ & \second{0.116} (0.060) & \second{0.163} (0.081) & \second{0.084} (0.098) & \second{0.177} (0.117) \\
\rowcolor{gray!30}
N$^3$POM & MPS & $\round{h_i}$ & 0.516 (0.044) & 0.527 (0.020) & 0.530 (0.040) & 0.525 (0.040) \\
\hline
POM & \EscVerb{polr} & $\round{h_i}$ & 0.516 (0.026) & 0.514 (0.017) & 0.514 (0.030) & 0.524 (0.034) \\ 
NPOM & (ridge) \EscVerb{oNet} & $\round{h_i}$ & 0.233 (0.037) & 0.265 (0.073) &   0.356 (0.030) & 2.572 (0.169) \\
NPOM & (elastic) \EscVerb{oNet} & $\round{h_i}$ & 0.243 (0.170) & 0.270 (0.166) & 0.215 (0.070) & 0.229 (0.101) \\
NPOM & (lasso) \EscVerb{oNet} & $\round{h_i}$ & 0.209 (0.151) & 0.266 (0.173) & 0.237 (0.085) & 0.223 (0.105) \\
NPOM$^{\dagger}$ & (ridge) \EscVerb{oNet} & $\round{h_i}$ & 0.253 (0.055) & 0.270 (0.070) & 0.418 (0.023) & 1.129 (0.080) \\
NPOM$^{\dagger}$ & (elastic) \EscVerb{oNet} & $\round{h_i}$ & 0.262 (0.161) & 0.274 (0.144) & 0.198 (0.093) & 0.209 (0.089) \\
NPOM$^{\dagger}$ & (lasso) \EscVerb{oNet} & $\round{h_i}$ & 0.261 (0.153) & 0.265 (0.179) & 0.206 (0.097) & 0.243 (0.120) \\
NPOM & \EscVerb{serp} & $\round{h_i}$ & 0.174 (0.066) & 0.204 (0.079) & 0.130 (0.074) & 0.186 (0.074) \\
\bottomrule
\end{tabular}
}
\end{table}

\section{Experiments on Real-World Datasets}
\label{sec:experiments_real}

In this section, we train N$^3$POM by leveraging real-world datasets. 
We show and interpret the results of the experiments for autoMPG6 and real-estate datasets. 
Also see Supplement~\ref{supp:additional_experiments} for the results of autoMPG8, boston-housing, concrete, and airfoil datasets. 
\verb|R| source codes to reproduce the experimental results are provided in \url{https://github.com/oknakfm/N3POM}.

\subsection{Real-world datasets}
\label{subsec:real-world_datasets}

We employ the following datasets collected from the UCI machine learning repository~\citep{dua2019UCI}.

\paragraph{autoMPG6 ($n=392, d=5$).} 
The autoMPG6 dataset consists of five covariates (``Displacement'' (continuous), ``Horse\_power'' (continuous), ``Weight'' (continuous), ``Acceleration'' (continuous), ``Model\_year'' (discrete)), and continuous response ``mpg''. ``mpg'' stands for miles per gallon, representing fuel efficiency.

\paragraph{real-estate ($n=413, d=3$).} 
Among the six covariates in the original real-estate dataset, we employ the three covariates 
``X2:house age'', ``X3:distance to the nearest MRT station'', and ``X4:number of convenience stores'', and  remove the following covariates: ``X1:transaction date'', ``X5:latitude'', and ``X6:longtitude''. 
They are renamed to ``House\_age'', ``Dist\_to\_station'', and ``Num\_of\_conv\_stores'' in our experiments. 
The response variable represents the house price of the unit area. Thus, we rename this variable as ``house price oua''. 
For meaningful computation, we remove the 271st column because its observed response is an outlier whose difference from the mean of the remaining responses is farther than six times the standard deviation.

\bigskip
For each dataset, observed covariates are standardized (centering and scaling). 
Responses are linearly transformed so that $\min_i h_i=1, \max_i h_i=10$ (i.e., we set $J=10$). After computing N$^3$POM, we recover the original response for the plot by applying the inverse linear transformation. 
See Supplement~\ref{supp:additional_experiments} for N$^3$POM applied to autoMPG8~($n=392,d=7$), boston-housing~($n=502,d=12$), concrete~($n=1030,d=8$), and airfoil~($n=1503,d=5$) datasets, and 
Supplement~\ref{supp:summary_of_datasets} for pairwise scatter plots of the observed covariates in these datasets.

\subsection{Experimental settings}
Initialization and optimization are the same as the experiments in Section~\ref{sec:experiments_synthetic}, while we employ $R=20$ with equal intervals $1=j_1<j_2<\cdots<j_R=J=10$ in real-world dataset experiments. 

First, we compute the \verb|serp| function by leveraging the rounded responses $\round{h_i}$. 
Initialized by this \verb|serp| output (see Supplement~\ref{supp:initialization}), we train the neural network by the MPS algorithm. 
For considering the randomness of choosing mini-batches, we compute 10 different stochastic optimization results (i.e., 10 different random seeds for stochastic optimization) in each plot.  
For comparison, we also plot the initial neural network (approximating the preliminarily computed \verb|serp| output) and POM coefficients trained using the \verb|polr| function.

\subsection{Results}

Estimated coefficients for each dataset are shown in Figures~\ref{fig:autoMPG6}--\ref{fig:boston-housing}. 
For interpretation, we inverse the sign of the estimated coefficients $\hat{\bs b}(u)$ because we have a probability that the response exceeds the threshold value $u$ as: 
\begin{align*}
    \text{logit}
    \left(
    \hat{\mathbb{P}}_{\cnpom}(H > u \mid X=\bs x)
    \right)
    &=
    \text{logit}
    \left(
    1-\hat{\mathbb{P}}_{\cnpom}(H \le u \mid X=\bs x) 
    \right)
    =
    \hat{r}(u)+\langle \hat{\bs s}(u),\bs x\rangle
\end{align*}
with $\hat{r}(u)=-\hat{a}(u)$ and $\hat{\bs s}(u) \, = \, (\hat{s}_1(u),\hat{s}_2(u),\ldots,\hat{s}_d(u)) \, = \, -\hat{\bs b}(u)$. Therefore, the larger $\hat{s}_k(u)=-\hat{b}_k(u)$ (corresponding to smaller $\hat{b}_k(u)$) indicates a larger response if the corresponding observed covariate is large.

\paragraph{autoMPG6} 
We consider the results of autoMPG6 ($n=392, d=5$) dataset shown in Figures \ref{fig:autoMPG6}. 
The coefficients $\hat{s}_k(u)$ of ``Displacement'', ``Horse\_power'', and ``Weight'' are negative. Therefore, taking higher displacement (or horse power/weight) indicates lower fuel efficiency. 
Their negative values are also decreasing along with the mpg. This descendence indicates that for cars with a higher mpg, the negative association between these covariates and the mpg is even stronger. 
The remaining covariates $\hat{s}_k(u)$ for ``Acceleration'' and ``Model\_year'' are increasing, which indicates that the relationship between these variables and fuel efficiency is more positive for more fuel-efficient cars. Particularly, the coefficient function for ``Model\_year'' is always positive, so the newness of the car implies better mpg for cars with all the range of mpg. 
Also see Figure~\ref{fig:autoMPG8} in Supplement~\ref{supp:additional_experiments} for the result of autoMPG8 dataset; the result is consistent with that of autoMPG6, and it indicates the robustness of the proposed N$^3$POM against the additional covariates.

\begin{figure}[!h]
\centering
\includegraphics[width=\textwidth]{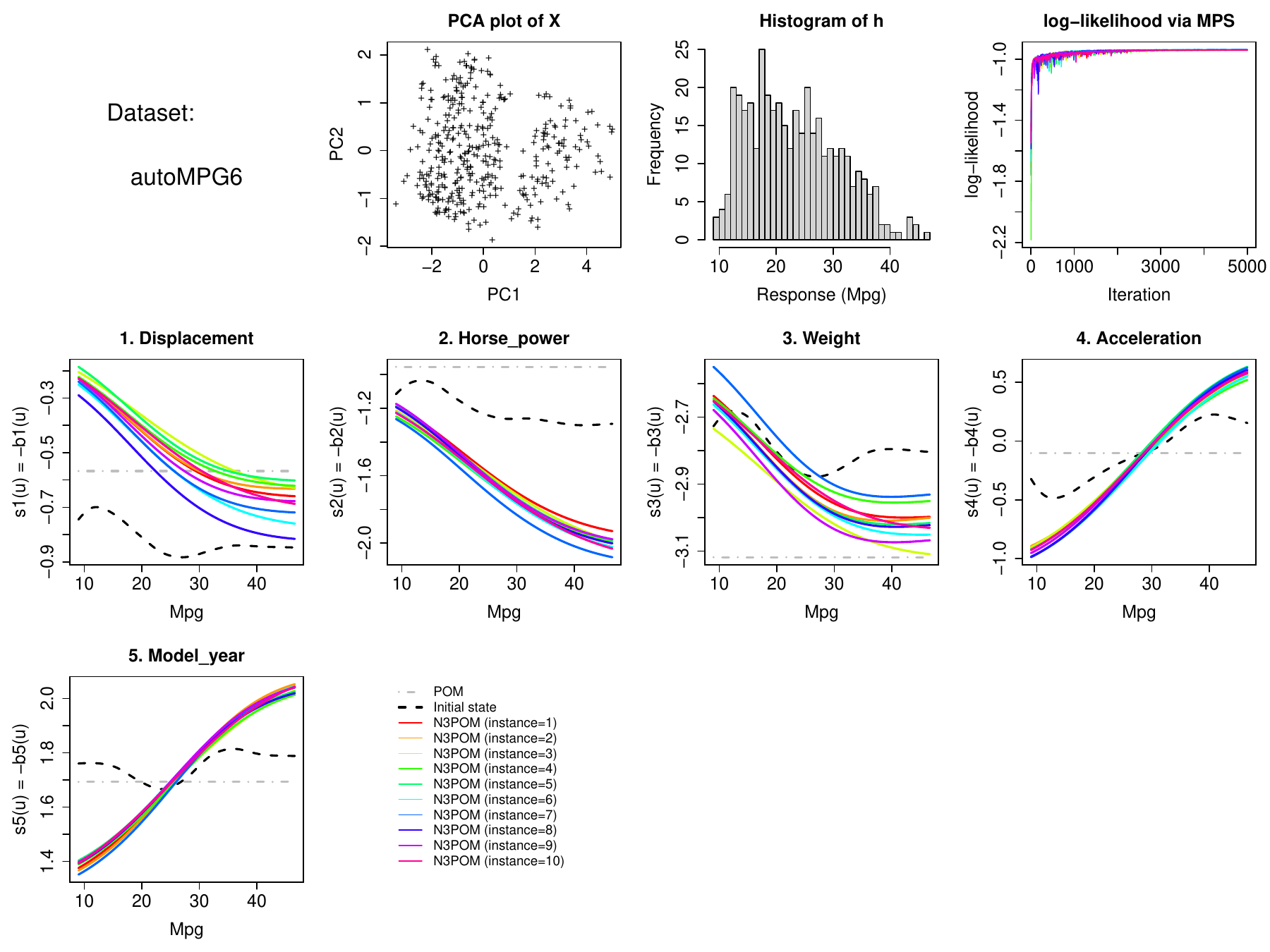}
\caption{autoMPG6 dataset experiment. 
The dashed line represents the coefficient function using the untrained (initial) NN output. 
Separate $10$ curves represent the coefficient functions with different random seeds for stochastic optimization. As these curves seem almost the same in all experiments, we can assert that the functions are estimated robustly against the stochastic optimization procedure. 
}
\label{fig:autoMPG6}
\end{figure}

\paragraph{real-estate.} 
The results of real-estate ($n=413, d=3$) dataset are shown in Figure \ref{fig:real-estate-main}(\subref{fig:real-estate}). 
For the second covariate, ``Dist\_to\_station'', the increasing distance adversely affects the house price. Moreover, the degree of the adverse effect increases for more expensive houses. 
The third covariate, ``Num\_of\_conv\_store'' positively affects the house price. However, the degree of the positive effect decreases for more expensive houses.  
For the first covariate, ``House\_age'', the age of the house adversely affects the house price for lower-price houses; however, this effect almost vanishes (as $\hat{s}_1(u)$ approaches $0$ as $u$ increases) for higher price houses. See Figure~\ref{fig:real-estate-main}(\subref{fig:real-estate-levels}) for the scatter plots between covariates and house price.  For instance, we can observe that the house age does not seem to adversely affect the house price for higher price houses, while it seems to have a slightly negative effect when considering lower-priced houses.

\begin{figure}[!h]
\centering
\begin{minipage}{\textwidth}
\centering
\includegraphics[width=\textwidth]{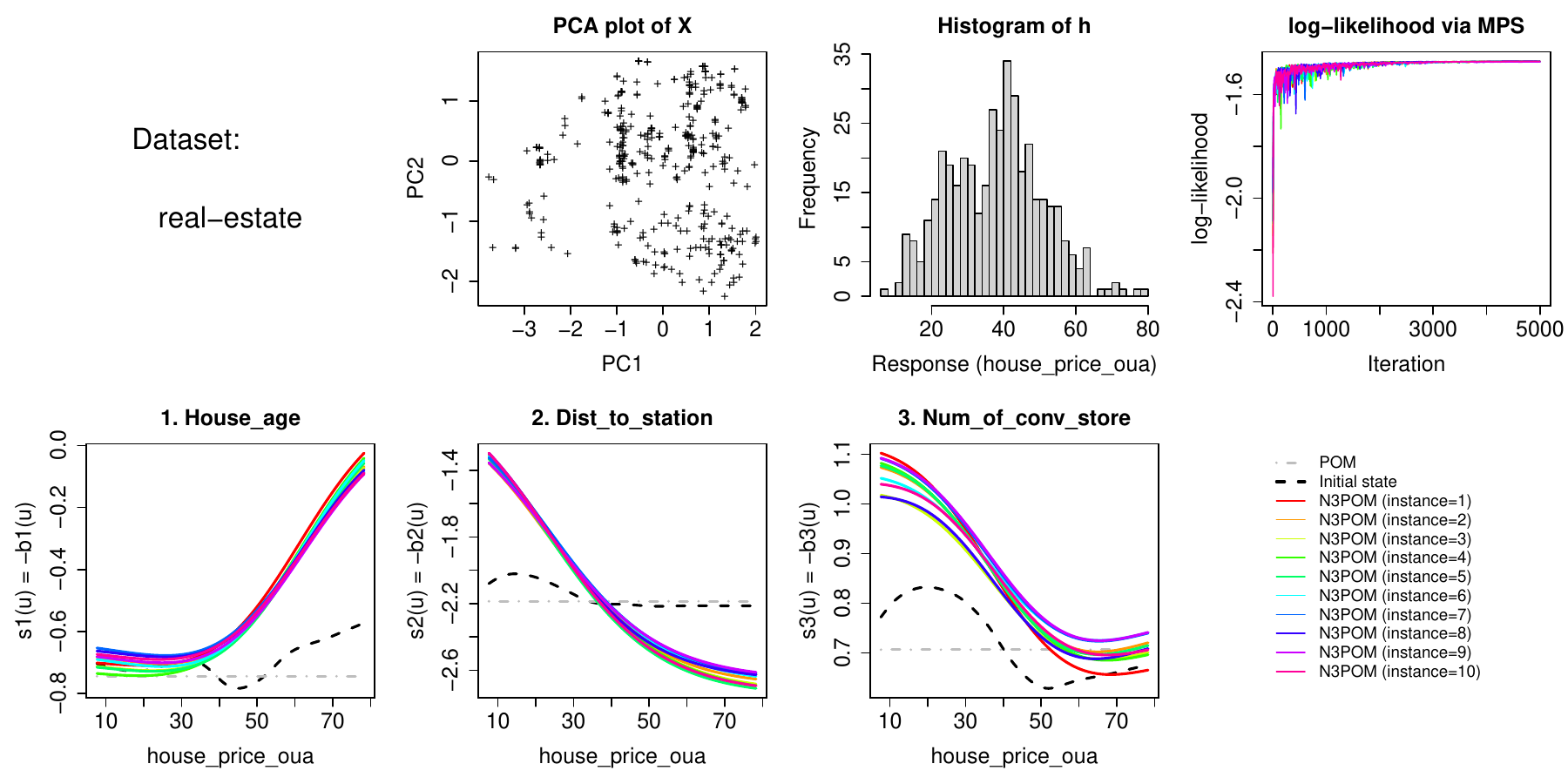}
\subcaption{Experimental results.}
\label{fig:real-estate}
\end{minipage} 
\hspace{1em}
\begin{minipage}{\textwidth}
\centering
\includegraphics[width=0.8\textwidth]{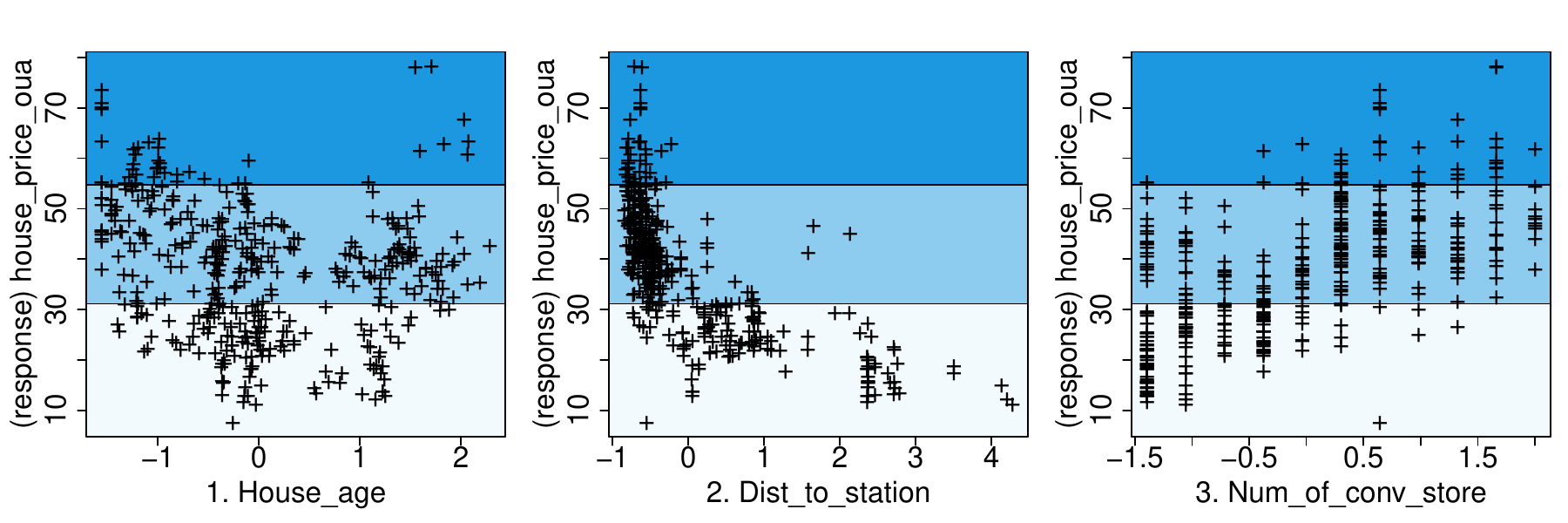}
\subcaption{Scatter plots between covariates (house age, distance to station, and number of stores; $x$-axis) and response (house price of unit area; $y$-axis). Higher price, middle price, and lower price houses are separately colored. 
}
\label{fig:real-estate-levels}
\end{minipage}
\caption{real-estate dataset.}
\label{fig:real-estate-main}
\end{figure}

\bigskip 
Please note that the descriptions of the datasets, along with the learning curves, are presented in the first row of Figures \ref{fig:autoMPG6} and \ref{fig:real-estate-main} respectively. 
The regression results of several additional datasets are also provided in Supplement~\ref{supp:additional_experiments}. 
Additionally, see Supplement~\ref{supp:summary_of_datasets} for pairwise scatter plots of all datasets used in this study, and 
Supplement~\ref{supp:marginal_effects} for the discussion on marginal effects~\citep{agresti2018simple}.

\section{Future Reseach Directions}
As a potential future research, it would be valuable to relax the sufficient condition~\eqref{eq:NN_weights_inequality} that ensures monotonicity.
As discussed in Section~\ref{subsec:monotonicity}, the current condition~\eqref{eq:NN_weights_inequality} is hard to be satisfied for higher-dimensional covariates.
Another possible direction is to develop statistical inference for N$^3$POM-based estimation.
For example, assessing the reliability of the coefficient functions would be a critical problem.
The final direction we would like to highlight is the application of (group) sparse penalties on the log-likelihood, which would encourage certain neural network weights to become zero. In this case, the neural network $b_k(u)$ may become a constant function in some cases, potentially helping to prevent overfitting of non-proportional models.

\newpage
\appendix

\noindent\textbf{{\Large Supplementary material:}} \\
\textbf{An interpretable neural network-based non-proportional odds model} \\
\textbf{for ordinal regression} \\
{A. Okuno (\textit{okuno@ism.ac.jp}) and K. Harada (\textit{haradak@tokyo-med.ac.jp}).}

\section{Gradient}
\label{supp:gradient}

With respect to the parameter $\bs \theta$, gradients $\nabla_{\bs \theta}f_u,\nabla_{\bs \theta}f^{[1]}_u$ of the prediction model $f_u$ and its weak derivative $f^{[1]}_u$ are shown in Supplement~\ref{supp:gradient_f}, \ref{supp:gradient_f1}, respectively; together with these gradients, we show the gradient of the log-likelihood in Supplement~\ref{supp:gradient_log-likelihood}.

\subsection{Gradient of the regression model}
\label{supp:gradient_f}

Let 
\[
    \trunc{z} 
    = 
    \begin{cases}
        0 & (z < 0) \\
        z & (z \in [0,1]) \\
        1 & (z > 1) \\
    \end{cases}, 
    \quad 
    z \in \mathbb{R}. 
\]
Considering the identity 
\[
    a(u) 
    =
    \phi 
    +
    \sum_{r=2}^{R} |\varphi_r| 
    \trunc{\frac{u-j_{r-1}}{j_r-j_{r-1}}},
\]
the gradient $\nabla_{\bs \theta} f_u(\bs x_i)$ is obtained element-wise as: 
for $\rdagger \in \{2,3,\ldots,R\},\kdagger \in [d],\ldagger \in [L]$ and $u \in \setU$, 
\begin{align*}
\frac{\partial}{\partial \phi} f_u(\bs x_i)
&=
1, \\
\frac{\partial}{\partial \varphi_{\rdagger}} f_u(\bs x_i)
&=
\sign(\varphi_{\rdagger})
\trunc{\frac{u-j_{\rdagger-1}}{j_{\rdagger}-j_{\rdagger-1}}}, \\
\frac{\partial}{\partial v_{\kdagger,\ldagger}^{(1)}} f_u(\bs x_i)
&=
    x_{i\kdagger}
    \frac{\partial}{\partial v_{\kdagger,\ldagger}^{(1)}}
    b_{\kdagger}(u) 
=
    x_{i\kdagger}
    w_{\kdagger,\ldagger}^{(2)}
    \rho^{[1]}( w_{\kdagger,\ldagger}^{(1)} u + v_{\kdagger,\ldagger}^{(1)} ), \\
\frac{\partial}{\partial v_{\kdagger}^{(2)}} f_u(\bs x_i)
&=
    x_{i\kdagger}
    \frac{\partial}{\partial v_{\kdagger,\ldagger}^{(2)}}
    b_{\kdagger}(u) 
=
    x_{i\kdagger},\\    
\frac{\partial}{\partial w_{\kdagger,\ldagger}^{(1)}} f_u(\bs x_i)
&=
    x_{i\kdagger}
    \frac{\partial}{\partial w_{\kdagger,\ldagger}^{(1)}}
    b_{\kdagger}(u) 
=
    x_{i\kdagger}
    w_{\kdagger,\ldagger}^{(2)} u
    \rho^{[1]}( w_{\kdagger,\ldagger}^{(1)} u + v_{\kdagger,\ldagger}^{(1)} ), \\
\frac{\partial}{\partial w_{\kdagger,\ldagger}^{(2)}} f_u(\bs x_i)
&=
    x_{i\kdagger}
    \frac{\partial}{\partial w_{\kdagger,\ldagger}^{(2)}}
    b_{\kdagger}(u) 
=
    x_{i\kdagger}
    \rho( w_{\kdagger,\ldagger}^{(1)} u + v_{\kdagger,\ldagger}^{(1)} ). 
\end{align*}

\subsection{Gradient of the weak derivative of the prediction model}
\label{supp:gradient_f1}

Considering the identities 
\[
a^{[1]}(u)
=
\sum_{r=2}^{R}
\mathbbm{1}(u \in \setU_{r-1})
\frac{|\varphi_r|}{j_{r}-j_{r-1}}, 
\quad 
b_k^{[1]}(u)
=
\sum_{\ell=1}^{L}
w_{k,\ell}^{(2)}
w_{k,\ell}^{(1)}
\rho^{[1]}(w_{k,\ell}^{(1)}u + v_{k,\ell}^{(1)}), 
\]
the gradient $\nabla_{\bs \theta} f_u^{[1]}(\bs x)$ of the weak derivative $f_u^{[1]}(\bs x)$ of the prediction model $f_u(\bs x)$ is obtained element-wise as follows: for $\rdagger \in \{2,3,\ldots,R\},\kdagger \in [d], \ldagger \in [L]$ and $u \in \setU$, we have 
\begin{align*}
\frac{\partial}{\partial \phi} f_u^{[1]}(\bs x_i)
&=
0, \\
\frac{\partial}{\partial \varphi_{\rdagger}} f_u^{[1]}(\bs x_i)
&=
\mathbbm{1}(u \in \setU_{\rdagger-1})
\frac{\sign(\varphi_{\rdagger})}{j_{\rdagger}-j_{\rdagger-1}}, \\
\frac{\partial}{\partial v_{\kdagger,\ldagger}^{(1)}} f_u^{[1]}(\bs x_i)
&=
    x_{i\kdagger}
    \frac{\partial}{\partial v_{\kdagger,\ldagger}^{(1)}}
    b_{\kdagger}^{[1]}(u) 
=
    x_{i\kdagger}
    w_{\kdagger,\ldagger}^{(2)}
    w_{\kdagger,\ldagger}^{(1)}
    \rho^{[2]}( w_{\kdagger,\ldagger}^{(1)} u + v_{\kdagger,\ldagger}^{(1)} ), \\
\frac{\partial}{\partial v_{\kdagger}^{(2)}} f_u^{[1]}(\bs x_i)
&=
    x_{i\kdagger}
    \frac{\partial}{\partial v_{\kdagger,\ldagger}^{(2)}}
    b_{\kdagger}^{[1]}(u) 
=
    0,\\    
\frac{\partial}{\partial w_{\kdagger,\ldagger}^{(1)}} f_u^{[1]}(\bs x_i)
&=
    x_{i\kdagger}
    \frac{\partial}{\partial w_{\kdagger,\ldagger}^{(1)}}
    b_{\kdagger}^{[1]}(u) 
=
    x_{i\kdagger}
    w_{\kdagger,\ldagger}^{(2)}
    w_{\kdagger,\ldagger}^{(1)}
    u
    \rho^{[2]}( w_{\kdagger,\ldagger}^{(1)} u + v_{\kdagger,\ldagger}^{(1)}) \\
    &\hspace{17em}
    +
    x_{i\kdagger}
    w_{\kdagger,\ldagger}^{(2)}
    \rho^{[1]}( w_{\kdagger,\ldagger}^{(1)} u + v_{\kdagger,\ldagger}^{(1)} ), \\
\frac{\partial}{\partial w_{\kdagger,\ldagger}^{(2)}} f_u^{[1]}(\bs x_i)
&=
    x_{i\kdagger}
    \frac{\partial}{\partial w_{\kdagger,\ldagger}^{(2)}}
    b_{\kdagger}^{[1]}(u) 
=
    x_{i\kdagger}
    w_{\kdagger,\ldagger}^{(1)}
    \rho^{[1]}( w_{\kdagger,\ldagger}^{(1)} u + v_{\kdagger,\ldagger}^{(1)} ). 
\end{align*}

\subsection{Gradient of the log-likelihood}
\label{supp:gradient_log-likelihood}

Using the conditional probability density function~\eqref{eq:conditional_pdf}, we get

\begin{align*}
\nabla_{\bs \theta} \ell_{\bs \zeta}(\bs \theta)
&=
\sum_{i=1}^{n} \zeta_i \nabla_{\bs \theta}\log q(h_i \mid \bs x_i) \\
&=
\sum_{i=1}^{n} \zeta_i \nabla_{\bs \theta}
\{
    \log \sigma^{[1]}(f_{h_i}(\bs x_i))
    +
    \log f_{h_i}^{[1]}(\bs x_i)
\} \\
&=
\sum_{i=1}^{n} \zeta_i 
\left\{
    \frac{\sigma^{[2]}(f_{h_i}(\bs x_i))}{\sigma^{[1]}(f_{h_i}(\bs x_i))}
    \nabla_{\bs \theta}f_{h_i}(\bs x_i)
    +
    \frac{1}{f_{h_i}^{[1]}(\bs x_i)}\nabla_{\bs \theta} f_{h_i}^{[1]}(\bs x_i)
\right\}.
\end{align*}

Together with the gradient of $f_u$ and $f_u^{[1]}$, the gradient of the log-likelihood is obtained element-wise as follows: 
For $\rdagger \in \{2,3,\ldots,R\},\kdagger \in [d], \ldagger \in [L]$, we have 
\begin{align*}
\frac{\partial}{\partial \phi} \ell_{\bs \zeta}(\bs \theta)
&=
\sum_{i=1}^{n} \zeta_i \frac{\sigma^{[2]}(f_{h_i}(\bs x_i))}{\sigma^{[1]}(f_{h_i}(\bs x_i))}, \\
\frac{\partial}{\partial \varphi_{\rdagger}} \ell_{\bs \zeta}(\bs \theta)
&=
\sign(\varphi_{\rdagger})
\sum_{i=1}^{n} \zeta_i \frac{\sigma^{[2]}(f_{h_i}(\bs x_i))}{\sigma^{[1]}(f_{h_i}(\bs x_i))}
\trunc{\frac{h_i-j_{\rdagger-1}}{j_{\rdagger}-j_{\rdagger-1}}} \\
&\hspace{3em}
+
\frac{\sign(\varphi_{\rdagger})}{j_{\rdagger}-j_{\rdagger-1}}
\sum_{i=1}^{n} \zeta_i
\mathbbm{1}(h_i \in \setU_{\rdagger-1})
\frac{1}{f_{h_i}^{[1]}(\bs x_i)}, \\
\frac{\partial}{\partial v_{\kdagger,\ldagger}^{(1)}} \ell_{\bs \zeta}(\bs \theta)
&=
\sum_{i=1}^{n} \zeta_i \frac{\sigma^{[2]}(f_{h_i}(\bs x_i))}{\sigma^{[1]}(f_{h_i}(\bs x_i))}
x_{i\kdagger}
    w_{\kdagger,\ldagger}^{(2)}
    \rho^{[1]}( w_{\kdagger,\ldagger}^{(1)} h_i + v_{\kdagger,\ldagger}^{(1)} ) \\
&\hspace{3em}+
\sum_{i=1}^{n} \zeta_i
\frac{1}{f_{h_i}^{[1]}(\bs x_i)}
 x_{i\kdagger}
    w_{\kdagger,\ldagger}^{(2)}
    w_{\kdagger,\ldagger}^{(1)}
    \rho^{[2]}( w_{\kdagger,\ldagger}^{(1)} h_i + v_{\kdagger,\ldagger}^{(1)} ), \\
\frac{\partial}{\partial v_{\kdagger}^{(2)}} \ell_{\bs \zeta}(\bs \theta)
&=
\sum_{i=1}^{n} \zeta_i \frac{\sigma^{[2]}(f_{h_i}(\bs x_i))}{\sigma^{[1]}(f_{h_i}(\bs x_i))}
x_{i\kdagger}, \\
\frac{\partial}{\partial w_{\kdagger,\ldagger}^{(1)}} \ell_{\bs \zeta}(\bs \theta)
&=
\sum_{i=1}^{n} \zeta_i \frac{\sigma^{[2]}(f_{h_i}(\bs x_i))}{\sigma^{[1]}(f_{h_i}(\bs x_i))}
x_{i\kdagger}
    w_{\kdagger,\ldagger}^{(2)} h_i
    \rho^{[1]}( w_{\kdagger,\ldagger}^{(1)} h_i + v_{\kdagger,\ldagger}^{(1)} ) \\
&\hspace{3em}+
\sum_{i=1}^{n} \zeta_i
\frac{1}{f_{h_i}^{[1]}(\bs x_i)}
\bigg\{
x_{i\kdagger}
    w_{\kdagger,\ldagger}^{(2)}
    w_{\kdagger,\ldagger}^{(1)}
    h_i
    \rho^{[2]}( w_{\kdagger,\ldagger}^{(1)} h_i + v_{\kdagger,\ldagger}^{(1)} ) \\
&\hspace{18em}
    +
    x_{i\kdagger}
    w_{\kdagger,\ldagger}^{(2)}
    \rho^{[1]}( w_{\kdagger,\ldagger}^{(1)} h_i + v_{\kdagger,\ldagger}^{(1)} )
\bigg\}, \\
\frac{\partial}{\partial w_{\kdagger,\ldagger}^{(2)}} \ell_{\bs \zeta}(\bs \theta)
&=
\sum_{i=1}^{n} \zeta_i \frac{\sigma^{[2]}(f_{h_i}(\bs x_i))}{\sigma^{[1]}(f_{h_i}(\bs x_i))}
x_{i\kdagger}
    \rho( w_{\kdagger,\ldagger}^{(1)} h_i + v_{\kdagger,\ldagger}^{(1)} ) \\
&\hspace{3em}+
\sum_{i=1}^{n} \zeta_i
\frac{1}{f_{h_i}^{[1]}(\bs x_i)}
x_{i\kdagger}
    w_{\kdagger,\ldagger}^{(1)}
    \rho^{[1]}( w_{\kdagger,\ldagger}^{(1)} h_i + v_{\kdagger,\ldagger}^{(1)} ). 
\end{align*}

\section{Initialization}
\label{supp:initialization}

With given parameter vectors $\hat{\bs \beta}_{j}=(\hat{\beta}_{j1},\hat{\beta}_{j2},\ldots,\hat{\beta}_{jd})^{\top}$ ($j=1,2,\ldots,J$), which are preliminarily computed using the existing algorithm for (discrete) NPOM, we initialize the neural network parameters to satisfy
\begin{align}
    b_k(j) \approx \hat{\beta}_{jk},
    \quad (j=1,2,\ldots,J; k=1,2,\ldots,d).
    \label{eq:distillation_equality}
\end{align}
In our implementation, we employ the vectors $\hat{\bs \beta}_1,\hat{\bs \beta}_2,\ldots,\hat{\bs \beta}_{J-1}$ computed by \verb|serp| package in \verb|R| language and specify $\hat{\bs \beta}_J=\hat{\bs \beta}_{J-1}+(\hat{\bs \beta}_{J-1}-\hat{\bs \beta}_{J-2})=2\hat{\bs \beta}_{J-1}-\hat{\bs \beta}_{J-2}$ formally.

To satisfy the equality~\eqref{eq:distillation_equality}, we employ an NN using sigmoid activation function $\rho(z)=1/(1+\exp(-z))$ and $L>d$; with a sufficiently large constant $T$ (e.g., $T=10$, satisfying $\rho(-T) \approx 0,\rho(+T) \approx 1$), we define
\begin{align*}
    v_k^{(2)}&= \frac{1}{J} \sum_{j=1}^{J}\hat{\beta}_{jk}, \\
     v_{k,\ell}^{(1)}&=\begin{cases}  -T \ell & (\ell \in \{1,2,\ldots,J\}) \\ 0 & (\text{Otherwise}) \end{cases}, \\
    w_{k,\ell}^{(1)}&=\begin{cases} T & (\ell \in \{1,2,\ldots,J\}) \\ 0 & (\text{Otherwise}) \end{cases}, \\
    w_{k,\ell}^{(2)}&=\begin{cases} 
    \dfrac{\hat{\beta}_{\ell k}-v_k^{(2)} - \sum_{\ell’=1}^{\ell-1}w_{k,\ell’}^{(2)}}{\rho(0)}
    &(\ell \in \{1,2,\ldots,J\}) \\ 0 &(\text{Otherwise}) \end{cases}, 
\end{align*}
for all $k \in \{1,2,\ldots,d\}$. 
Next, we have
\begin{align*}
b_k(j)
&=
\sum_{\ell=1}^{L} w_{k,\ell}^{(2)} \rho \left(
    w_{k,\ell}^{(1)}j+v_{k,\ell}^{(1)}
\right) + v_k^{(2)} \\
&=
\sum_{\ell=1}^{J} w_{k,\ell}^{(2)} \rho \left(
    T(j-\ell)
\right) + v_k^{(2)} \\
&\approx 
\sum_{\ell=1}^{J} w_{k,\ell}^{(2)}
\left\{
    \rho(-\infty) \mathbbm{1}(j <\ell)
    +
    \rho(0) \mathbbm{1}(\ell=j) 
    +
    \rho(\infty) \mathbbm{1}(\ell<j)
\right\} + v_k^{(2)} \\
&=
w_{k,j}^{(2)} \rho(0) 
+
\sum_{\ell=1}^{j-1} w_{k,\ell}^{(2)} + v_k^{(2)} \\&=
\dfrac{\hat{\beta}_{jk}-v_k^{(2)} - \sum_{\ell=1}^{j-1}w_{k,\ell}^{(2)}}{\rho(0)} \rho(0)
+
\sum_{\ell=1}^{j-1} w_{k,\ell}^{(2)} + v_k^{(2)} \\
&=
\hat{\beta}_{jk}.
\end{align*}

For more stability in the NN training in our implementation, the non-zero weights in $\{w_{k,\ell}^{(2)}\}$ divided by $T=\max_{t \in \mathbb{N}}\{L/J \ge t\}$ are duplicated $T$ times. Furthermore, the non-zero weights in $\{w_{k,\ell}^{(1)}\},\{v_{k,\ell}^{(1)}\}$ are also duplicated $T$ times. I.i.d. standard normal random numbers are added to the remaining zero-weights. 
Following this duplication, the weights are more uniformly distributed compared to the setting in which only a few non-zero weights exist. Accordingly, the subsequent neural network training is expected to be more stable.

\section{Proofs}
\label{supp:proofs}

\begin{proof}[Proof of Proposition~\ref{prop:constant_f}]
We employ a proof by contradiction. 
Given that $\setU$ is a compact set, $s:=\esssup_{u \in [1,J]}a^{[1]}(u)>0$ holds. 
Suppose there exists $(u',\bs x')\in \setU \times \mathbb{R}^d$ such that 
$|\nabla_u \langle \bs b(u'),\bs x' \rangle|>0$. 
Thus, we may assume $t_{u',\bs x'}:=\nabla_u \langle \bs b(u'),\bs x' \rangle>0$ without loss of generality (consider $-\bs x'$ if $\nabla_u \langle \bs b(u'),\bs x' \rangle<0$). 
Next, by taking $L:=(s/t_{u',\bs x'})+1$, we get
\[
    \nabla_u f_{u'}(-L\bs x')
    <
    s-L\nabla_u \langle \bs b(u'),\bs x' \rangle 
    <
    s - (\{s/t_{u',\bs x'}\}+1)t_{u',\bs x'}
    =
    -t_{u',\bs x'}
    <
    0
\]
almost surely. 
At point $-L\bs x'$, $f_u(-L\bs x)$ is decreasing with respect to $u$, over a sufficiently small open ball around $u'$. Thus, a contradiction is derived. Accordingly, $|\nabla_u \langle \bs b(u),\bs x \rangle|=0$ holds for all $(u,\bs x) \in \setU \in \mathbb{R}^d$. 
Together with the continuity of the function $\bs b$, $\bs b(u)$ is a constant function. 
\end{proof}

\begin{proof}[Proof of Proposition~\ref{prop:monotonicity}]
Given that the weak derivative is obtained as
\begin{align*}
\langle \bs b^{[1]}(u),\bs x \rangle
&=
\sum_{k=1}^{d} x_{k} b_k^{[1]}(u)
=
\sum_{k=1}^{d} x_{k} 
\sum_{\ell=1}^{L} w_{k,\ell}^{(2)} w_{k,\ell}^{(1)} 
\rho^{[1]} (w_{k,\ell}^{(1)}u + v_{k,\ell}^{(1)}),
\end{align*}
The Cauchy-Schwarz inequality proves an inequality
\begin{align}
\left| 
    \langle \bs b^{[1]}(u),\bs x \rangle
\right|
&\le 
    \sqrt{\sum_{k=1}^{d} x_{k}^2}
    \sqrt{
        \sum_{k=1}^{d}
        \left\{
            \sum_{\ell=1}^{L} w_{k,\ell}^{(2)} w_{k,\ell}^{(1)} 
            \rho^{[1]} (w_{k,\ell}^{(1)}u + v_{k,\ell}^{(1)})
        \right\}^2
    } \nonumber \\
&< 
    \xupper \cdot \rho^{[1]}_{\infty} \cdot 
    \sqrt{
        \sum_{k=1}^{d}
        \left\{
            \sum_{\ell=1}^{L} |w_{k,\ell}^{(2)} w_{k,\ell}^{(1)}|
        \right\}^2
    }
    \label{eq:cauchy-schwarz}
\end{align}
for $\bs x=(x_1,x_2,\ldots,x_d) \in \setX_2(\eta)$. This inequality indicates that
\begin{align*}
     f_u^{[1]}(\bs x)
=
    a^{[1]}(u)
    +
    \langle \bs b^{[1]}(u),\bs x \rangle
\overset{\eqref{eq:cauchy-schwarz}}{\ge}  
    \min_{r=2,3,\ldots,R-1}s_{r-1}
    -
    \xupper \cdot \rho^{[1]}_{\infty} \cdot 
    \sqrt{
        \sum_{k=1}^{d}
        \left\{
            \sum_{\ell=1}^{L} |w_{k,\ell}^{(2)} w_{k,\ell}^{(1)}|
        \right\}^2
    }
\overset{\eqref{eq:NN_weights_inequality}}{\ge} 0,
\end{align*}
where the non-negativity of the right-most side is obtained based on the inequality \eqref{eq:NN_weights_inequality}. 
Given that $f^{[1]}_u$ is a weak derivative of the function $f_u$, $f_u$ is non-decreasing. 
Given that the above calculation holds for all $\bs x$, the continuity of $f_u$ proves the assertion.  
\end{proof}

\section{Probability of truncation}
\label{app:truncation_probability}

We first evaluate a probability $p(\bs x) := \mathbb{P}(H \notin [1,J] \mid X=\bs x)$ with the setting $m_1=0.05,m_2=-0.05$ considered in Table~\ref{table:synthetic_setting1}: 
\begin{align*}
p(\bs x)
&=
1-\mathbb{P}(1 \le H \le J \mid X=\bs x) \\
&=
1-\{\mathbb{P}(H \le J \mid X=\bs x) - \mathbb{P}(H \le 1 \mid X=\bs x)\} \\
&=
1-\{\sigma(a_*(J)+\langle \bs b_*(J),\bs x\rangle) - \sigma(a_*(1)+\langle \bs b_*(1),\bs x \rangle)\} \\
&=
\sigma(-\{a_*(J)+\langle \bs b_*(J),\bs x\rangle\}) + \sigma(a_*(1)+\langle \bs b_*(1),\bs x \rangle).
\end{align*}
Substituting $a_*(u)=2u-9,\bs b_*(u)=(-1+m_1 u^2,1+m_2 u^2),J=7$, and $m_1=0.05,m_2=-0.05$ to the above formula leads to 
\begin{align*}
p(\bs x)
=
\sigma(- \{5+1.45(x_1-x_2)\})
-
\sigma(-7-0.95(x_1-x_2)). 
\end{align*}
Considering $x_1-x_2=r(\cos \theta-\sin\theta)=r(\cos\theta-\cos(\pi/2-\theta))=2r\cos(\pi/4)\cos(\theta-\pi/4)=2^{-1/2}r\cos(\theta-\pi/4)$, we have 
\begin{align*}
p(\bs x)
&=
\sigma(- 5 - 2^{-1/2} \times 1.45r \cos \psi\})
-
\sigma(- 7 - 2^{-1/2} \times 0.95r \cos \psi\}) 
\end{align*}
with $\psi:=\theta-\pi/4$. Therefore, we finally obtain the probability 
\[
    \mathbb{P}(H \notin [1,J])
    =
   \frac{1}{2\pi} \iint p(\bs x) \diff r \diff \psi
   =
   0.0063\cdots
\]
by conducting the Monte-Carlo integration.

\section{Additional experiments: synthetic datasets}
\label{supp:additional_experiments_synthetic}

We first conduct the additional experiments with the same setting as Section~\ref{sec:experiments_synthetic}, while at least one of $b_{*1}(u)=-1+m_1u^2,b_{*2}(u)=1+m_2u^2$ is a constant function, that is, either $m_1=0$ or $m_2=0$.
See Table~\ref{table:synthetic_setting1_supp} for the results. 
Similarly to the case of response-dependent coefficients (i.e., $m_1,m_2 \ne 0$), the N$^3$POM trained using the discretized response demonstrates a similar score to that of POM using \verb|polr| package (though we initialized the neural network using \verb|serp| package). 
For the constant coefficients (corresponding to $m_1=0$ or $m_2=0$), the POM and N$^3$POM trained using the discrete responses demonstrate the best and second-best scores. 
This is because the POM assumes the constant coefficients, which is why the estimation variance is significantly smaller than that of the more flexible models, NPOM and N$^3$POM.

\begin{table}[H]
\centering
\caption{ 
Results of the MSE experiments on synthetic datasets, where the observed covariates are generated by the setting $x_i=(r_i \cos \theta_i, r_i \sin \theta_i), r_i \sim U([0,1]),\theta_i \sim U([0,2\pi))$. At least one of $b_{*1}(u)=-1+m_1u^2,b_{*2}(u)=1+m_2 u^2$ is a constant function, that is, either $m_1=0$ or $m_2=0$. For the 20 experiments for each setting, the robust average and standard deviation (shown in parenthesis) are computed. The best score is \best{bolded and blue-colored}, while the second best score is \second{bolded and red-colored}}
\label{table:synthetic_setting1_supp}
\scalebox{0.85}{
\begin{tabular}{rrr|cc|cc}
\toprule 
\multirow{2}{*}{Model} & \multirow{2}{*}{Optimizer} & \multirow{2}{*}{Response} & \multicolumn{2}{c|}{$(m_1,m_2)=(0.05,0)$} & \multicolumn{2}{c}{$(m_1,m_2)=(0,0)$} \\
 & & & MSE($\hat{b}_1$) & MSE($\hat{b}_2$) & MSE($\hat{b}_1$) & MSE($\hat{b}_2$) \\
\hline
\rowcolor{gray!30}
N$^3$POM & MPS & $h_i$ & \best{0.083} (0.164) & 0.075 (0.048) & 0.046 (0.034) & 0.052 (0.060) \\
\rowcolor{gray!30}
N$^3$POM & MPS & $\cut{\round{h_i}}$ & \second{0.110} (0.129) & 0.041 (0.042) & 0.035 (0.048) & 0.036 (0.059) \\
\rowcolor{gray!30}
N$^3$POM & MPS & $\round{h_i}$ & 0.526 (0.037) & \second{0.007} (0.020) & \best{0.004} (0.016) & \second{0.011} (0.019) \\
POM & \EscVerb{polr} & $\round{h_i}$ & 0.513 (0.029) & \best{0.002} (0.023) & \second{0.008} (0.018) & \best{0.002} (0.019) \\
NPOM & (ridge) \EscVerb{oNet} & $\round{h_i}$ & 0.310 (0.053) & 0.533 (0.063) & 0.559 (0.044) & 0.561 (0.044) \\
NPOM & (elastic) \EscVerb{oNet} & $\round{h_i}$ & 0.221 (0.081) & 0.178 (0.145) & 0.286 (0.124) & 0.192 (0.131) \\ 
NPOM & (lasso) \EscVerb{oNet} & $\round{h_i}$ & 0.211 (0.098) &  0.198 (0.161) & 0.321 (0.127) & 0.228 (0.129) \\
NPOM$^{\dagger}$ & (ridge) \EscVerb{oNet} & $\round{h_i}$ & 0.388 (0.034) & 0.117 (0.034) & 0.121 (0.037) & 0.125 (0.039) \\
NPOM$^{\dagger}$ & (elastic) \EscVerb{oNet} & $\round{h_i}$ & 0.391 (0.153) & 0.022 (0.057) & 0.013 (0.033) & 0.028 (0.022) \\
NPOM$^{\dagger}$ & (lasso) \EscVerb{oNet} & $\round{h_i}$ & 0.277 (0.155) & 0.034 (0.032) & 0.011 (0.034) & 0.028 (0.040) \\
NPOM & \EscVerb{serp} & $\round{h_i}$ & 0.186 (0.077) & 0.027 (0.033) & 0.011 (0.019) & 0.020 (0.037) \\
\bottomrule
\end{tabular}
}
\end{table}

We also conduct the experiments with different covariates: we generate $x_{ij} \sim \text{Beta}(0.5,0.5)$ ($j=1,2$) i.i.d. uniformly and randomly. $\text{Beta}(t_1,t_2)$ denotes Beta distribution, whose density is proportional to $x^{t_1-1}(1-x)^{t_2-1}$. 
Experimental results are summarized in Table~\ref{table:synthetic_setting2}. 
Overall, the tendency is the same as in the aforementioned settings; N$^3$POM trained using the observed continuous response $h_i$ demonstrates the best scores for estimating almost all non-constant coefficients. Moreover, the POM and N$^3$POM trained using the discrete response $\round{h_i}$ demonstrate the best scores for estimating the constant coefficients.

\begin{table}[H]
\centering
\caption{Results of the MSE experiments on synthetic datasets, where the observed covariates are $x_{ij} \sim \text{Beta}(0.5,0.5)$. 
Among the 20 experiments for each setting, the robust average and standard deviation (shown in parenthesis) are computed. The best score is \best{bolded and blue-colored}, while the second best score is \second{bolded and red-colored}}
\label{table:synthetic_setting2}
\begin{subtable}{\linewidth}
\centering
\subcaption{Both coefficients $b_{*1}(u)=-1+m_1u^2,b_{*2}(u)=1+m_2u^2$ are response-dependent.}
\label{table:synthetic_setting2_nonconstant}
\scalebox{0.85}{
\begin{tabular}{rrr|cc|cc}
\toprule 
\multirow{2}{*}{Model} & \multirow{2}{*}{Optimizer} & \multirow{2}{*}{Response} & \multicolumn{2}{c|}{$(m_1,m_2)=(0.05,-0.05)$} & \multicolumn{2}{c}{$(m_1,m_2)=(0.05,0.05)$} \\
 & & & MSE($\hat{b}_1$) & MSE($\hat{b}_2$) & MSE($\hat{b}_1$) & MSE($\hat{b}_2$) \\
\hline
\rowcolor{gray!30}
N$^3$POM & MPS & $h_i$ & \best{0.179} (0.143) & \second{0.265} (0.206) & \second{0.200} (0.092) & \best{0.300} (0.202) \\
\rowcolor{gray!30}
N$^3$POM & MPS & $\round{h_i}$ & 0.250 (0.135) & 0.297 (0.148) & \best{0.186} (0.112) & \second{0.305} (0.164) \\
\rowcolor{gray!30}
N$^3$POM & MPS & $\cut{\round{h_i}}$ & 0.559 (0.055) & 0.599 (0.069) & 0.524 (0.024) & 0.563 (0.049) \\
POM  & \EscVerb{polr} & $\round{h_i}$ & 0.522 (0.033) & 0.549 (0.045) & 0.521 (0.034) & 0.553 (0.037)\\
NPOM & (ridge) \EscVerb{oNet} & $\round{h_i}$ & 0.251 (0.070) & 0.278 (0.049) & 0.414 (0.043) & 3.239 (0.075) \\
NPOM & (elastic) \EscVerb{oNet} & $\round{h_i}$ & 0.481 (0.182) & 0.461 (0.187) & 0.350 (0.242) & 1.386 (0.780)  \\
NPOM & (lasso) \EscVerb{oNet} & $\round{h_i}$ & 0.482 (0.255) & 0.439 (0.219) & 0.416 (0.418) & 1.502 (0.906) \\ 
NPOM$^{\dagger}$ & (ridge) \EscVerb{oNet} & $\round{h_i}$ & 0.284 (0.085) & 0.318 (0.084) & 0.450 (0.026) & 1.456 (0.104) \\
NPOM$^{\dagger}$ & (elastic) \EscVerb{oNet} & $\round{h_i}$ & 0.477 (0.148) & 0.364 (0.155) & 0.443 (0.12) & 0.571 (0.081) \\
NPOM$^{\dagger}$ & (lasso) \EscVerb{oNet} & $\round{h_i}$ & 0.497 (0.253) & 0.419 (0.223) & 0.451 (0.293) & 0.552 (0.094) \\
NPOM & \EscVerb{serp} & $\round{h_i}$ & \second{0.207} (0.098) & \best{0.213} (0.078) & 0.277 (0.111) & 0.319 (0.118) \\
\bottomrule
\end{tabular}
}
\end{subtable}
\begin{subtable}{\linewidth}
\centering
\caption{At least one of $b_{*1}(u)=-1+m_1u^2,b_{*2}(u)=1+m_2 u^2$ is a constant function.}
\label{table:synthetic_setting2_constant}
\scalebox{0.85}{
\begin{tabular}{rrr|cc|cc}
\toprule 
\multirow{2}{*}{Model} & \multirow{2}{*}{Optimizer} & \multirow{2}{*}{Response} & \multicolumn{2}{c|}{$(m_1,m_2)=(0.05,0)$} & \multicolumn{2}{c}{$(m_1,m_2)=(0,0)$} \\
 & & & MSE($\hat{b}_1$) & MSE($\hat{b}_2$) & MSE($\hat{b}_1$) & MSE($\hat{b}_2$) \\
\hline
\rowcolor{gray!30}
N$^3$POM & MPS & $h_i$ & \best{0.195} (0.161) & 0.036 (0.099) & 0.050 (0.082) & 0.067 (0.117) \\ 
\rowcolor{gray!30}
N$^3$POM & MPS & $\round{h_i}$ & 0.293 (0.132) & 0.055 (0.051) & 0.056 (0.034) & 0.045 (0.052) \\
\rowcolor{gray!30}
N$^3$POM & MPS & $\cut{\round{h_i}}$ & 0.518 (0.043) & \second{0.026} (0.041) & \second{0.040} (0.031) & \best{0.015} (0.056) \\
POM & \EscVerb{polr} & $\round{h_i}$ & 0.517 (0.020) & \best{0.021} (0.032) & \best{0.023} (0.021) & \second{0.016} (0.033) \\
NPOM & (ridge) \EscVerb{oNet} & $\round{h_i}$ & 0.384 (0.065) & 0.556 (0.068) & 0.507 (0.074) & 0.554 (0.079) \\
NPOM & (elastic) \EscVerb{oNet} & $\round{h_i}$ & 0.438 (0.200) & 0.387 (0.201) & 0.203 (0.153) & 0.183 (0.198) \\ 
NPOM & (lasso) \EscVerb{oNet} & $\round{h_i}$ & 0.482 (0.467) & 0.507 (0.361) & 0.169 (0.236) & 0.228 (0.203) \\
NPOM$^{\dagger}$ & (ridge) \EscVerb{oNet} & $\round{h_i}$ & 0.409 (0.050) & 0.151 (0.059) & 0.087 (0.067) & 0.129 (0.064) \\
NPOM$^{\dagger}$ & (elastic) \EscVerb{oNet} & $\round{h_i}$ & 0.482 (0.063) & 0.130 (0.123) & 0.037 (0.051) & 0.059 (0.064) \\
NPOM$^{\dagger}$ & (lasso) \EscVerb{oNet} & $\round{h_i}$ & 0.472 (0.058) & 0.118 (0.181) & \second{0.040} (0.060) & 0.052 (0.069) \\
NPOM & \EscVerb{serp} & $\round{h_i}$ & \second{0.278} (0.119) & 0.054 (0.052) & 0.041 (0.037) & 0.046 (0.053) \\
\bottomrule
\end{tabular}
}
\end{subtable}
\end{table}

\section{Additional experiments: real-world datasets}
\label{supp:additional_experiments}

In addition to the autoMPG6, autoMPG8, and real-estate datasets used in the main body of the study, we collected boston-housing, concrete, and airfoil datasets from the UCI machine learning repository~\citep{dua2019UCI}. 
We train the N$^3$POM by leveraging these datasets described in the following section and their results are shown in Figures~\ref{fig:boston-housing}--\ref{fig:airfoil}. 


\paragraph[autoMPG8]{autoMPG8 ($n=392, d=7$).} 
The autoMPG8 ($n=392,d=7$) datasets consists of 7 covariates: 5 covariates (``Displacement'', ``Horse\_power'', ``Weight'', ``Acceleration'', and ``Model\_year'') shared with the autoMPG6 datasets, and the remaining 2 covariates (``Cylinders'' and ``Origin'').

\paragraph[boston-housing]{boston-housing ($n=506, d=12$).} 
The boston-housing dataset consists of 12 covariates (``crim'' (continuous), ``zn'' (continuous), ``indus'' (continuous), ``chas'' (binary), ``nox'' (continuous), ``rm'' (continuous), ``age'' (continuous), ``dis'' (continuous), ``rad'' (discrete), ``tax'' (continuous), ``ptratio'' (continuous), ``lstat'' (continuous)) and a continuous response (``medv'') representing the housing price in boston. 
We preliminarily removed the ``black'' row for fairness. 
Detailed descriptions of the covariates and the response are as follows: 
\textbf{crim:} per capita crime rate by town, 
\textbf{zn:} proportion of residential land zoned for lots over 25,000 sq.ft, 
\textbf{indus:} proportion of non-retail business acres per town, 
\textbf{chas:} Charles River dummy variable (=1 the house is located next to the river; 0 otherwise), 
\textbf{nox:} nitric oxides concentration (parts per 10 million), 
\textbf{rm:} average number of rooms per dwelling, 
\textbf{age:} proportion of owner-occupied units built prior to 1940, 
\textbf{dis:} weighted distances to five Boston employment centers , 
\textbf{rad:} index of accessibility to radial highways, 
\textbf{tax:} full-value property-tax rate per $\$10,000$, 
\textbf{ptraito:} pupil-teacher ratio by town, 
\textbf{lstat:} lower status of the population. 
(response) \textbf{medv:} Median value of owner-occupied homes in $\$$1000's.

\paragraph[concrete]{concrete ($n=1030,d=8$).} 
The concrete dataset consists of eight covariates (``Cement'' (continuous), ``BlastFurnaceSlag'' (continuous), ``FlyAsh'' (continuous), ``Water'' (continuous), ``Superplasticizer'' (continuous), ``CoarseAggregate'' (continuous), ``FineAggregate'' (continuous), ``Age''(continuous)) and a continuous response (``ConcreteCompressiveStrength''). 

\paragraph[airfoil]{airfoil ($n=1503,d=5$).}
The airfoil dataset consists of $5$ covariates (
frequency in hertz ``freq'' (continuous), 
angle of attack in degrees ``angle'' (continuous),
chord length in meters ``chord'' (continuous), 
free-stream velocity in meters ``velocity'' (continuous), 
suction side displacement thickness in meters ``disp. thickness'' (continuous) 
and a continuous response, ``sound pressure'' representing the scaled sound pressure level in decibels.



\begin{figure}[!p]
\centering
\includegraphics[width=\textwidth]{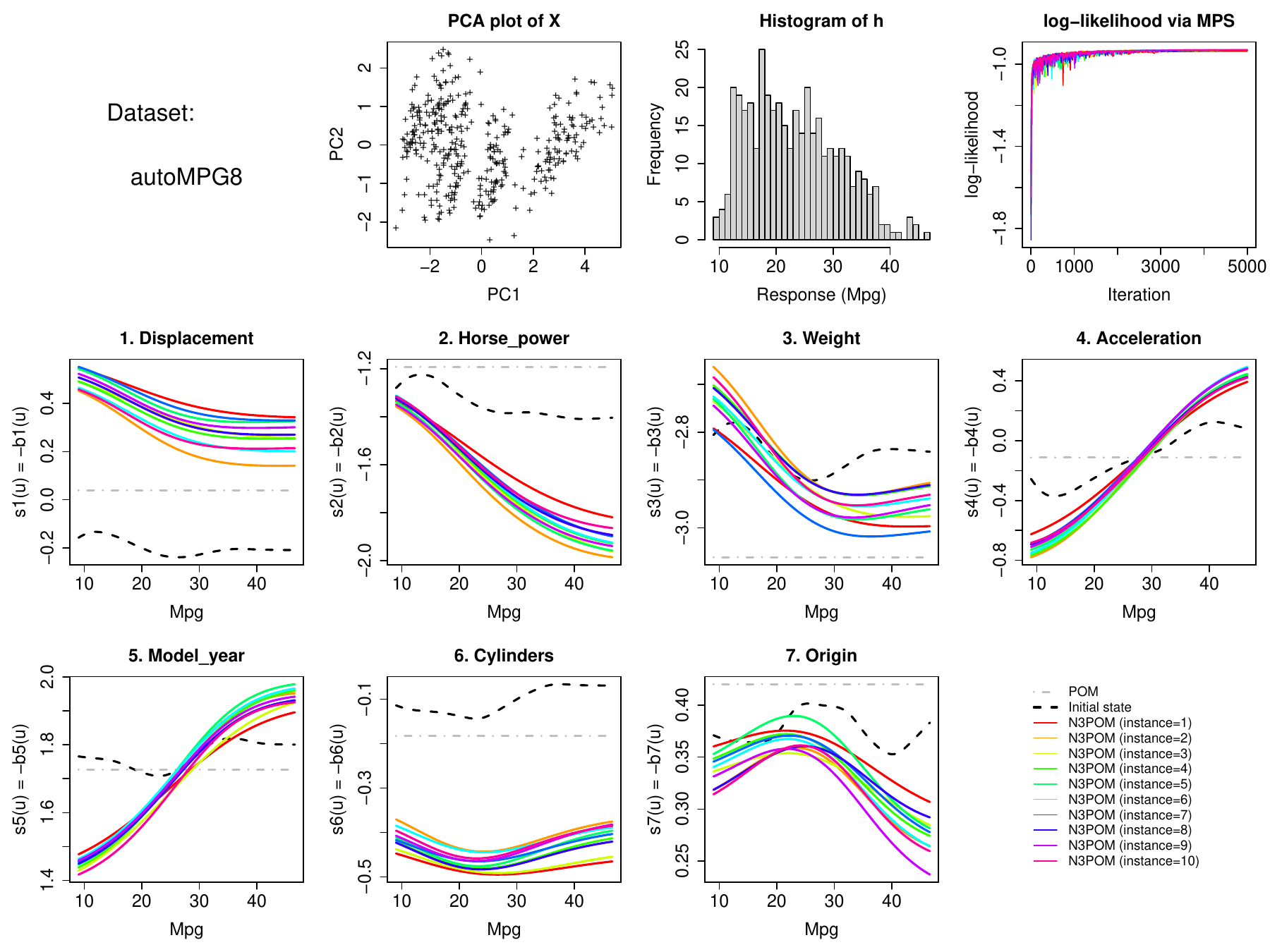}
\caption{autoMPG8 dataset experiment.}
\label{fig:autoMPG8}
\end{figure}

\begin{figure}[!p]
\centering 
\includegraphics[width=\textwidth]{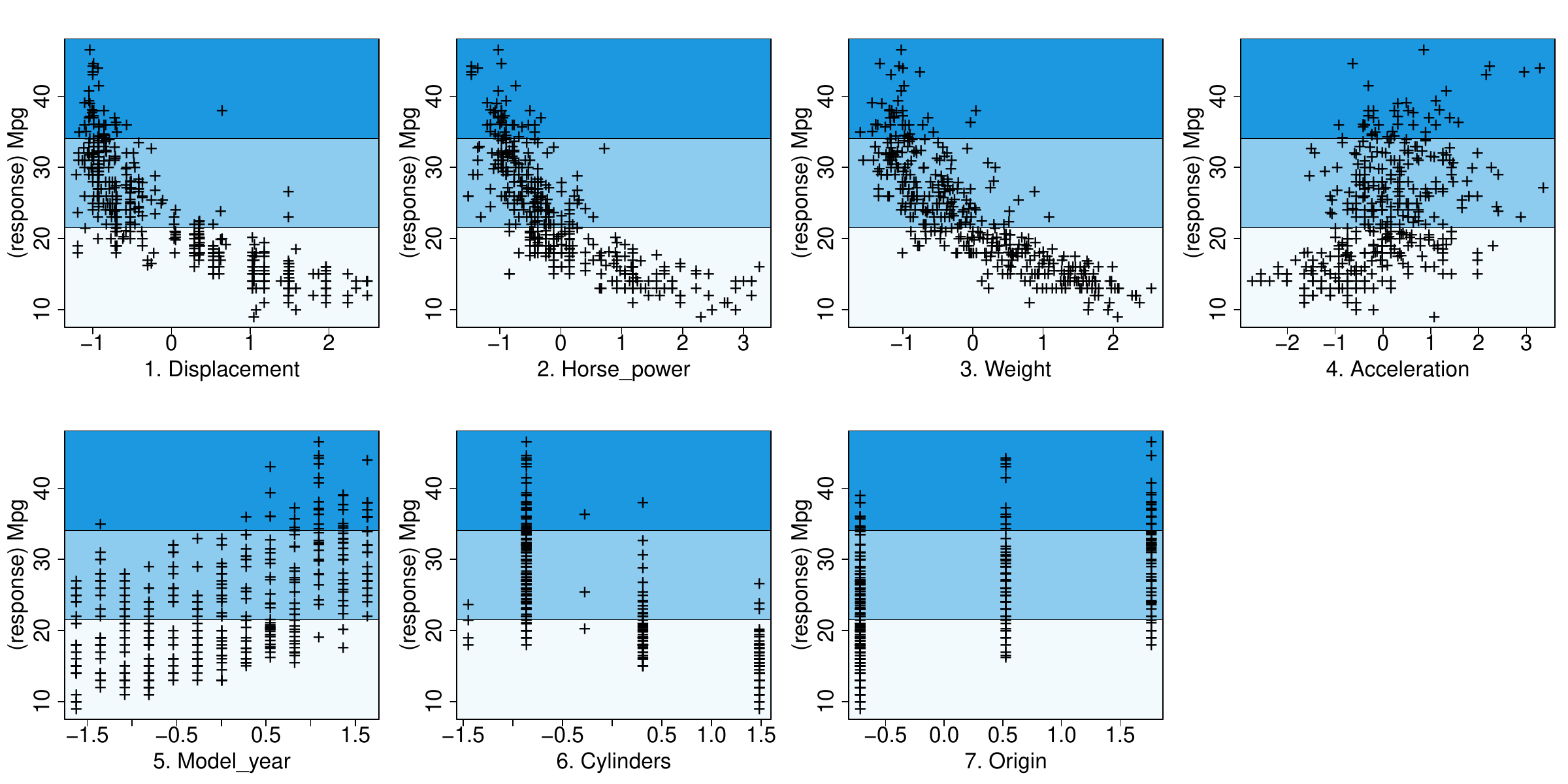}
\caption{Scatter plots for autoMPG8 dataset.}
\end{figure}

\begin{figure}[!p]
\centering
\includegraphics[width=\textwidth]{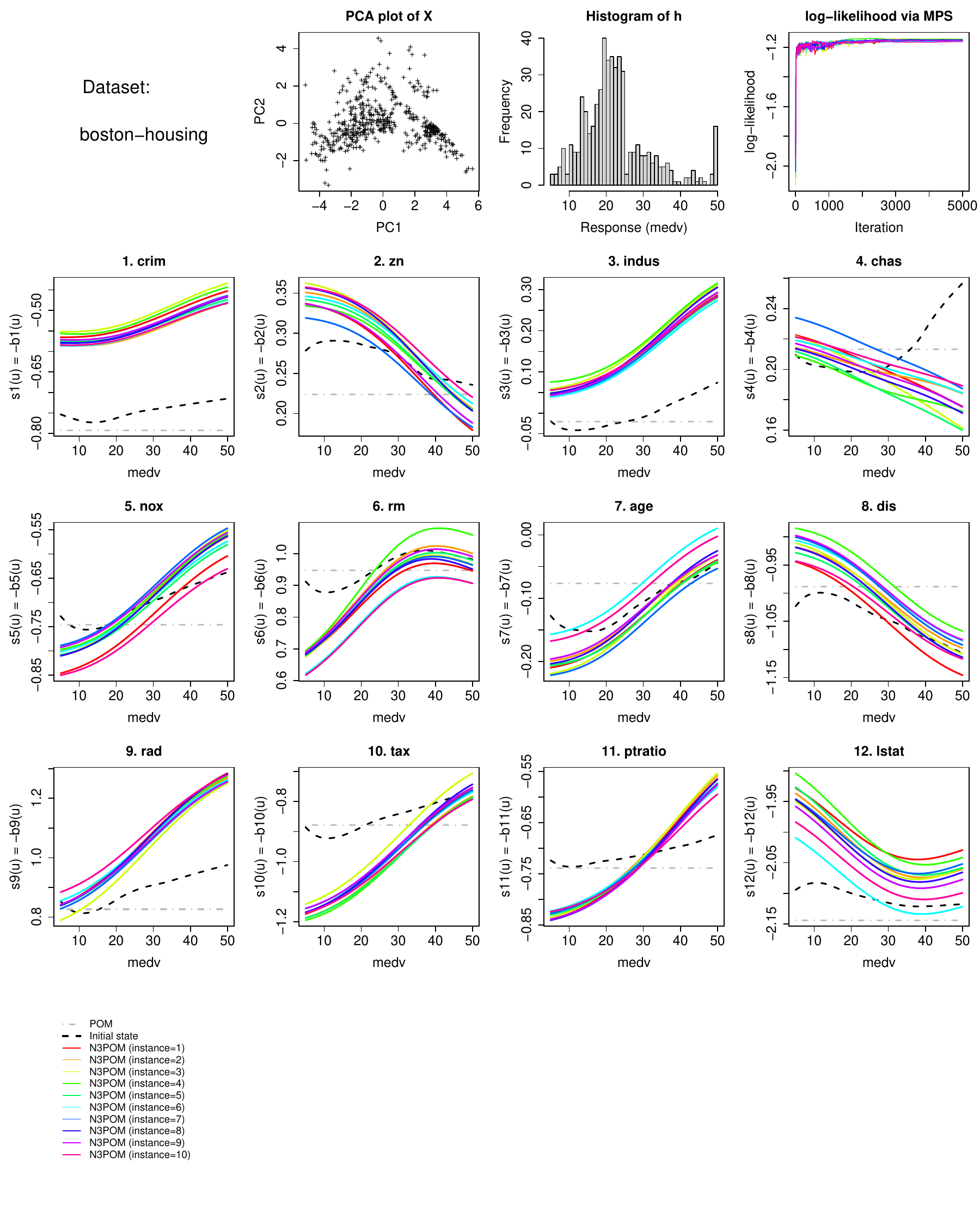}
\vspace{-5em}
\caption{Boston-housing dataset experiment.}
\label{fig:boston-housing}
\end{figure}

\begin{figure}[!p]
\centering 
\includegraphics[width=\textwidth]{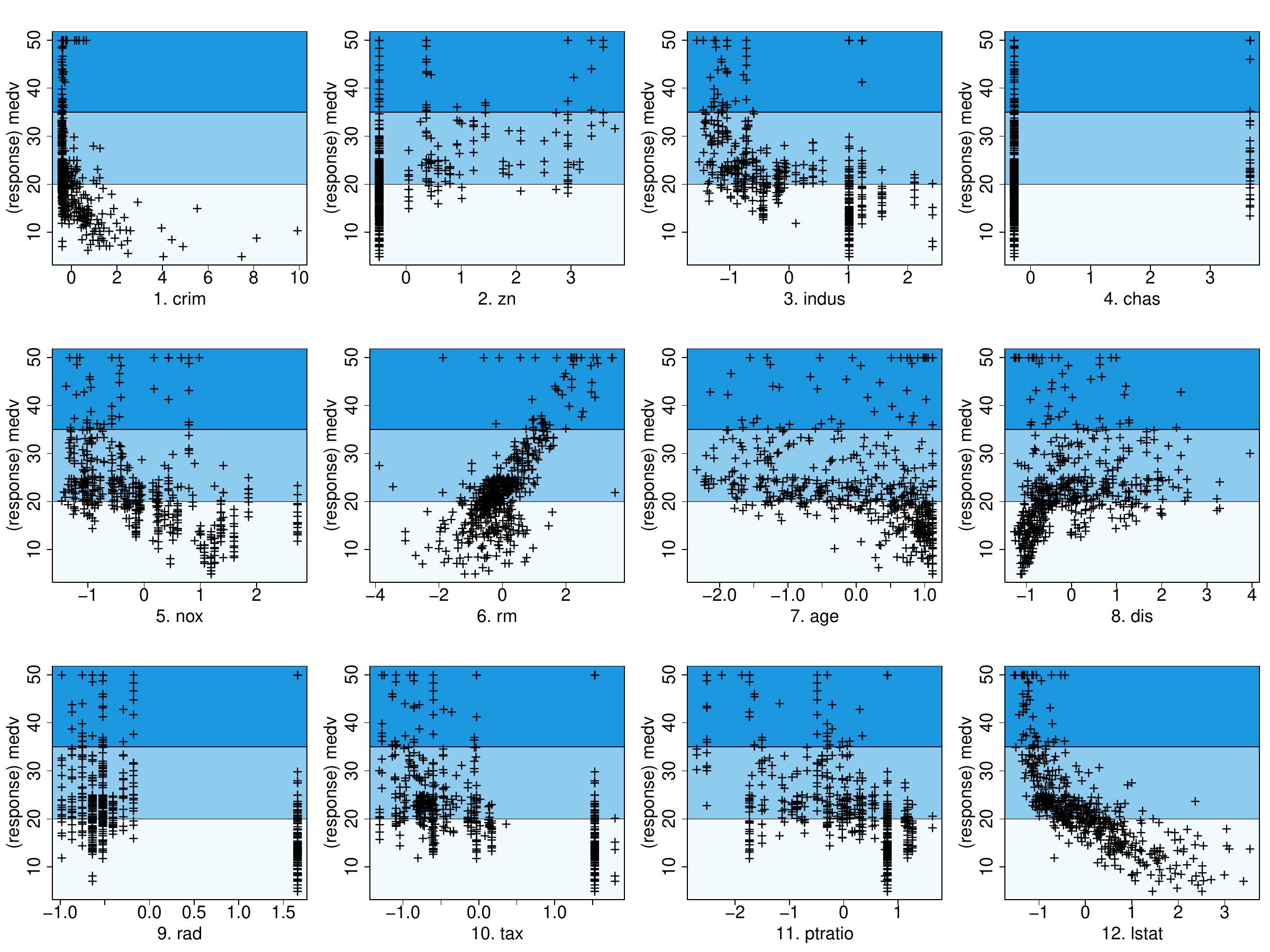}
\caption{Scatter plots for Boston-housing dataset.}
\end{figure}

\begin{figure}[!p]
\centering
\includegraphics[width=\textwidth]{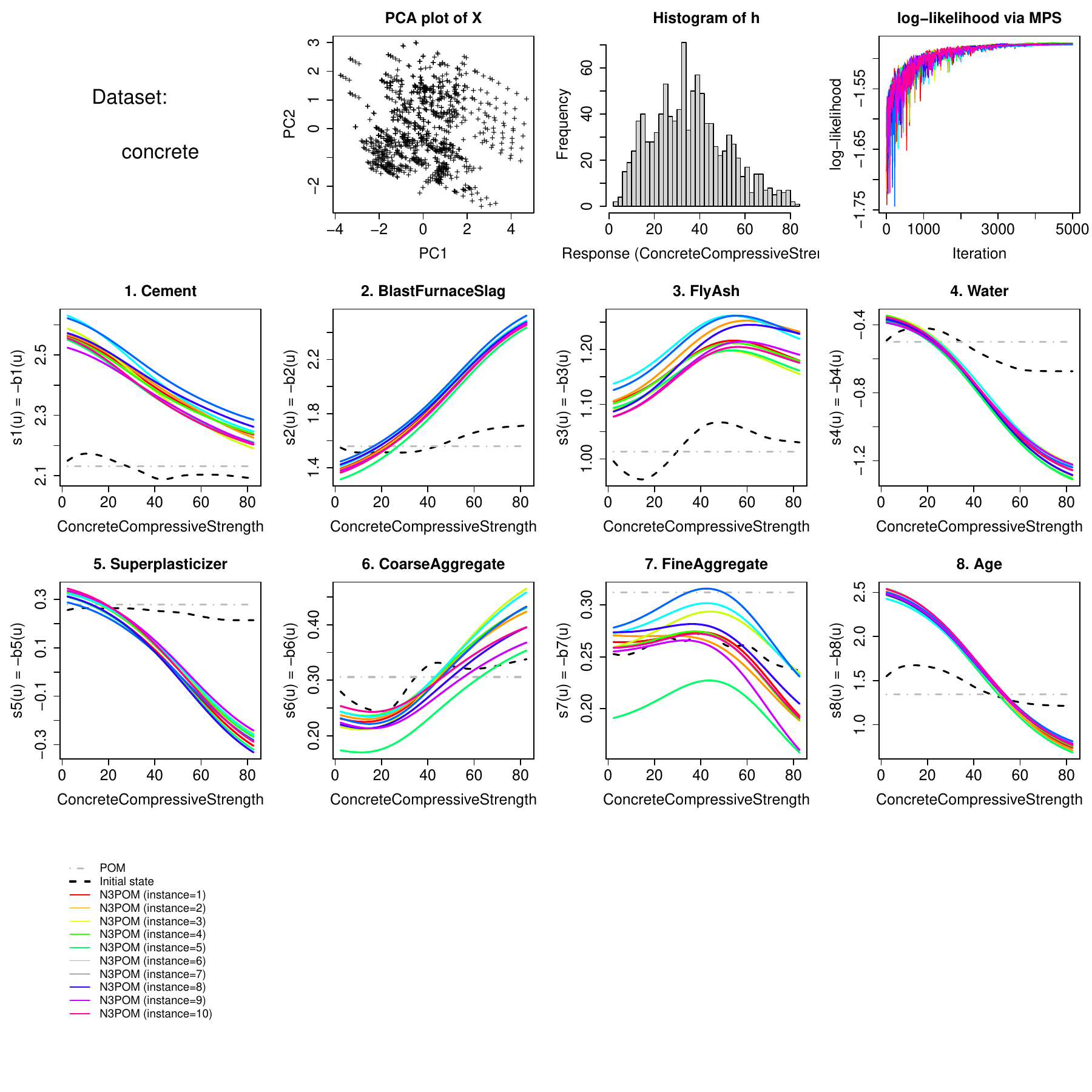}
\caption{concrete dataset experiment.}
\label{fig:concrete}
\end{figure}

\begin{figure}[!p]
\centering 
\includegraphics[width=\textwidth]{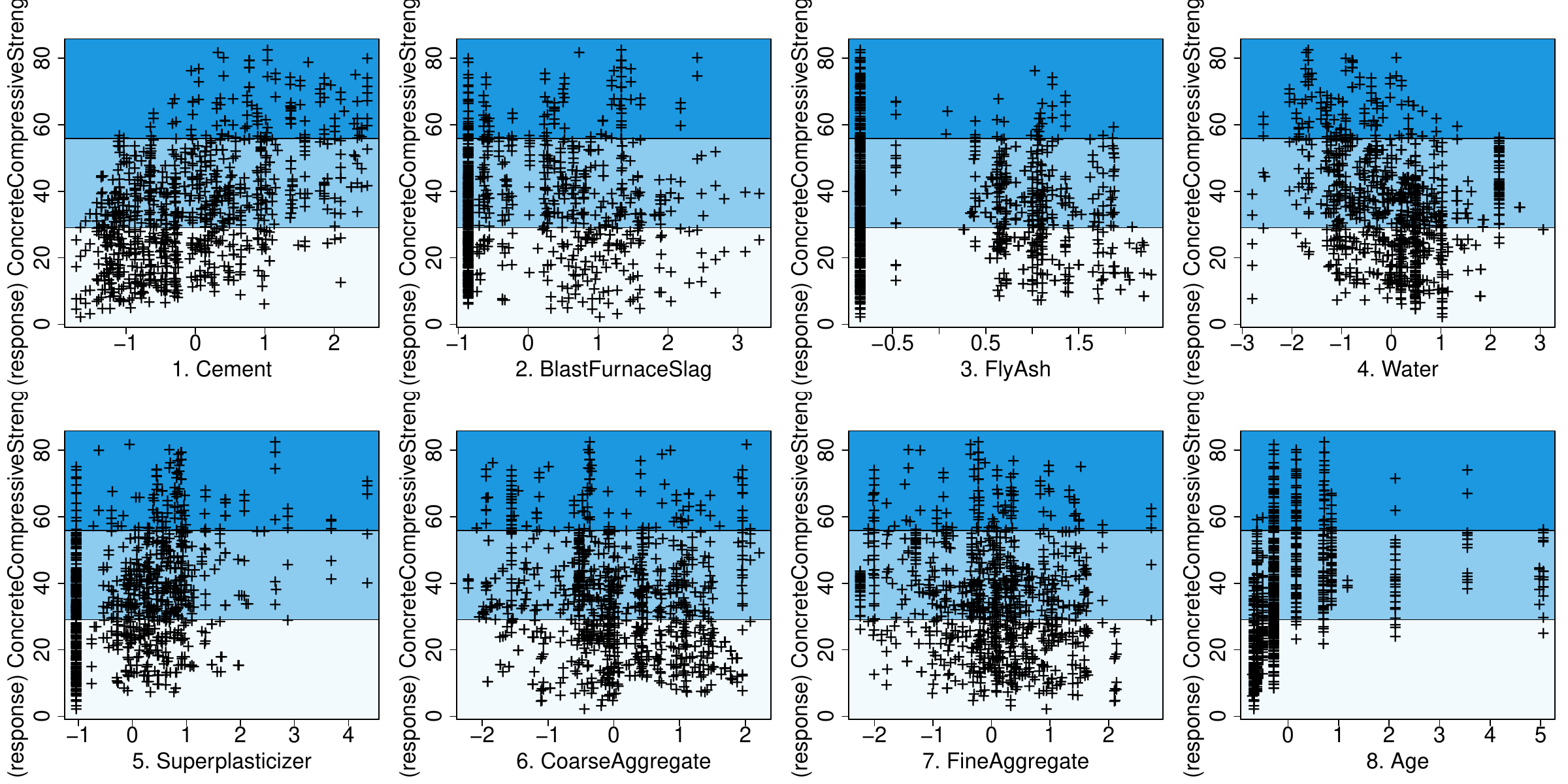}
\caption{Scatter plots for concrete dataset.}
\end{figure}

\begin{figure}[!p]
\centering
\includegraphics[width=\textwidth]{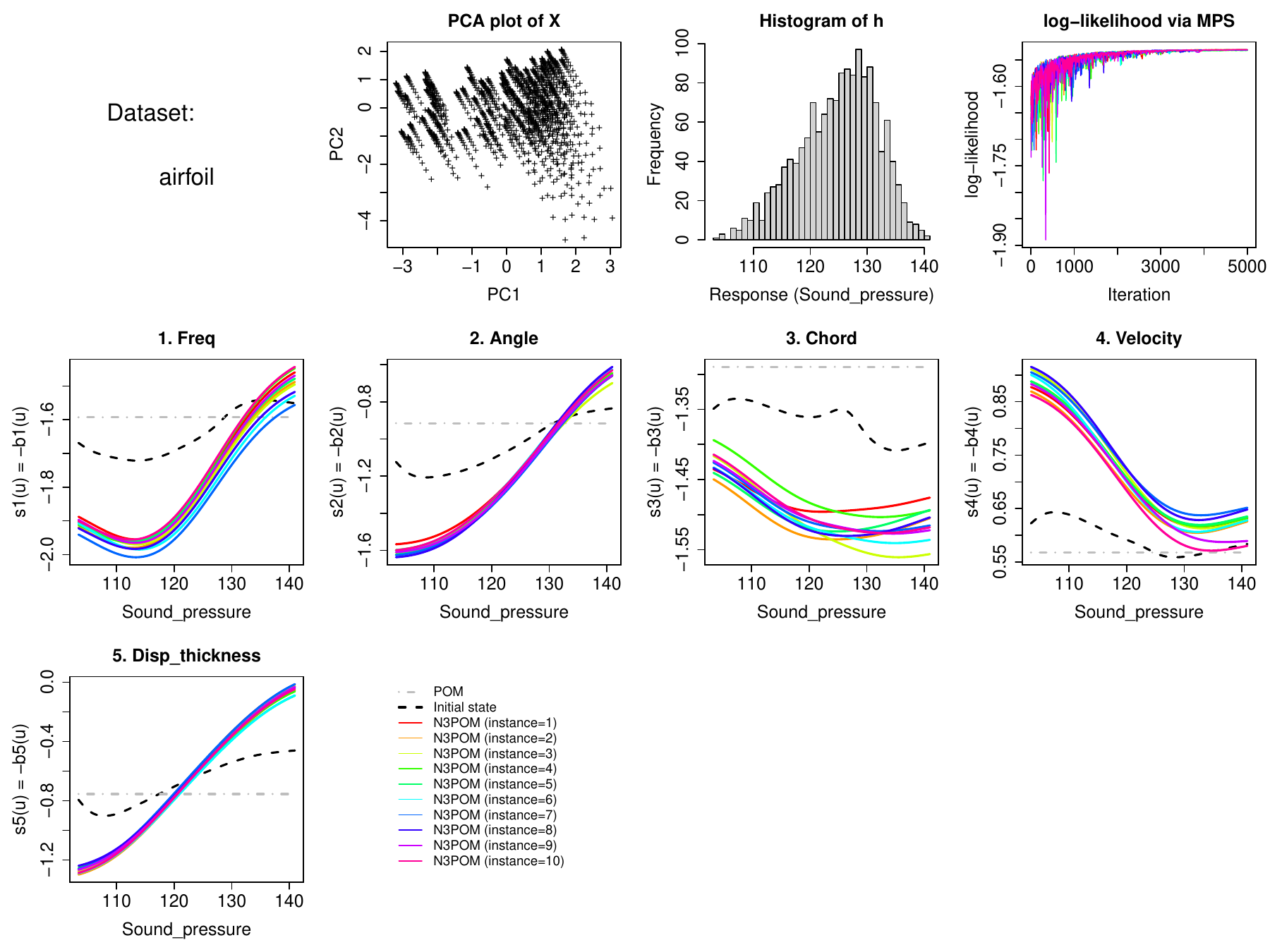}
\caption{airfoil dataset experiment.}
\label{fig:airfoil}
\end{figure}

\begin{figure}[!p]
\centering 
\includegraphics[width=\textwidth]{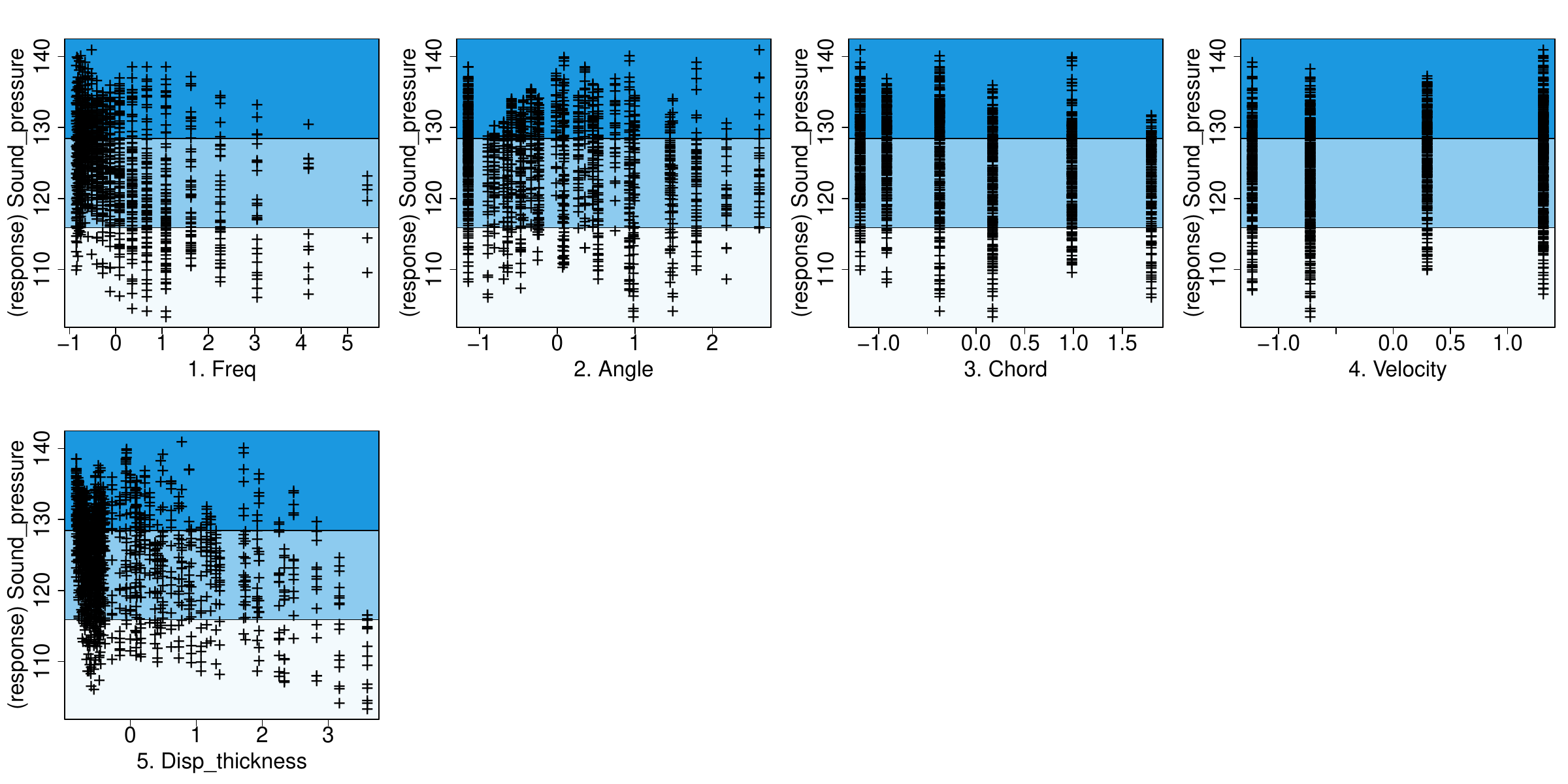}
\caption{Scatter plots for airfoil dataset.}
\end{figure}

\clearpage

\section{Summary of datasets.}
\label{supp:summary_of_datasets}

The covariates and response in autoMPG8, real-estate, boston-housing, concrete, airfoil, and cycle powerplant datasets are summarized in Figures~\ref{summary:autompg8}--\ref{summary:airfoil}. 
These plots are generated by \verb|pairs.panels| function in \verb|psych| package\footnote{\url{https://cran.r-project.org/package=psych}}. 
Therein, the scatter plots for each pair of covariates and their correlation coefficients are provided. 
We omit the plot for autoMPG6 because it is completely subsumed in autoMPG8. 
Note that the covariates and responses are preliminarily standardized by following the procedure described in Section~\ref{subsec:real-world_datasets}. 

\begin{figure}[!ht]
\centering 
\includegraphics[width=\textwidth]{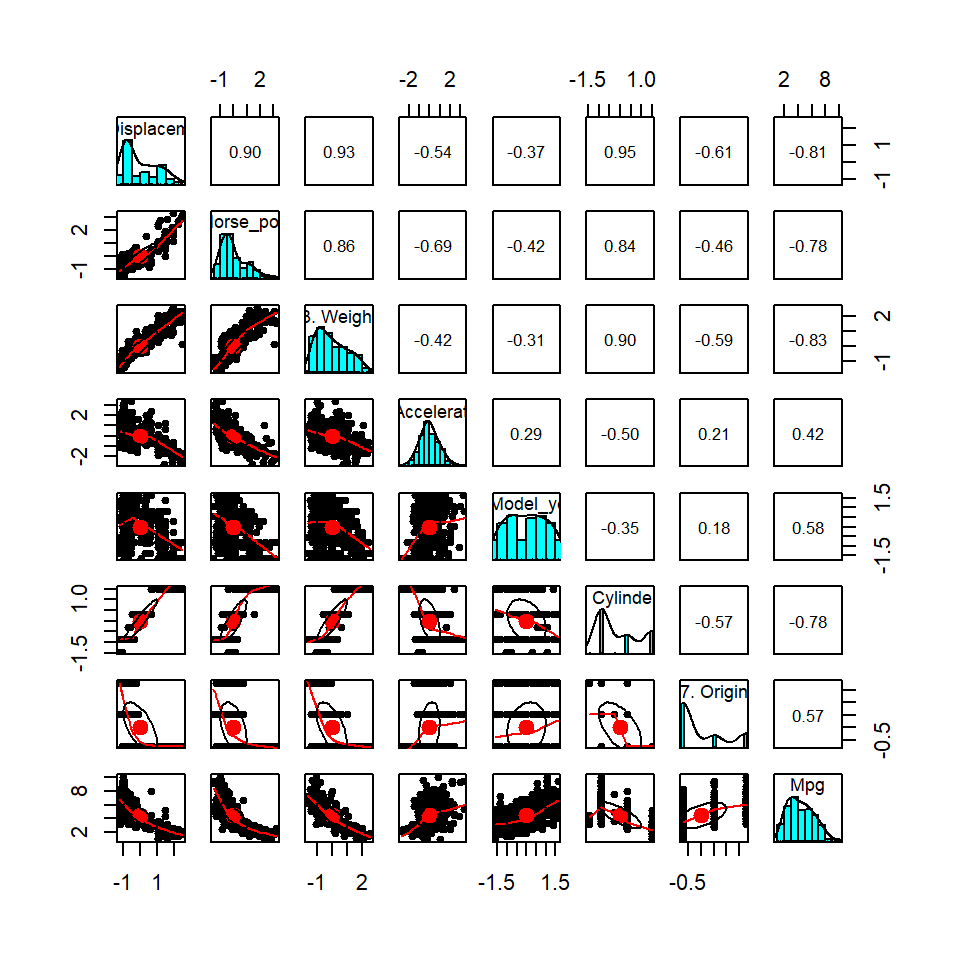}
\caption{Summary of autoMPG8 dataset (subsuming autoMPG6 dataset).}
\label{summary:autompg8}
\end{figure}

\begin{figure}[!ht]
\centering 
\includegraphics[width=0.8\textwidth]{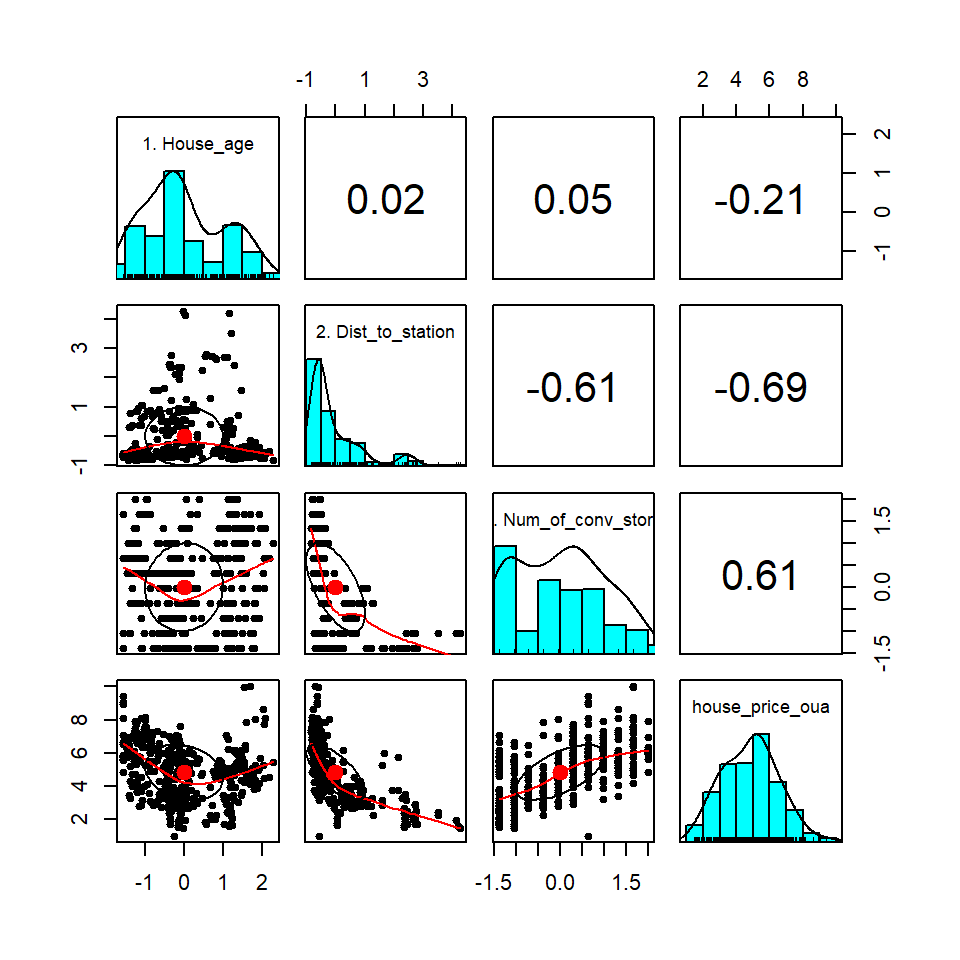}
\caption{Summary of the real-estate dataset.}
\label{summary:real-estate}
\end{figure}

\begin{figure}[!ht]
\centering 
\includegraphics[width=\textwidth]{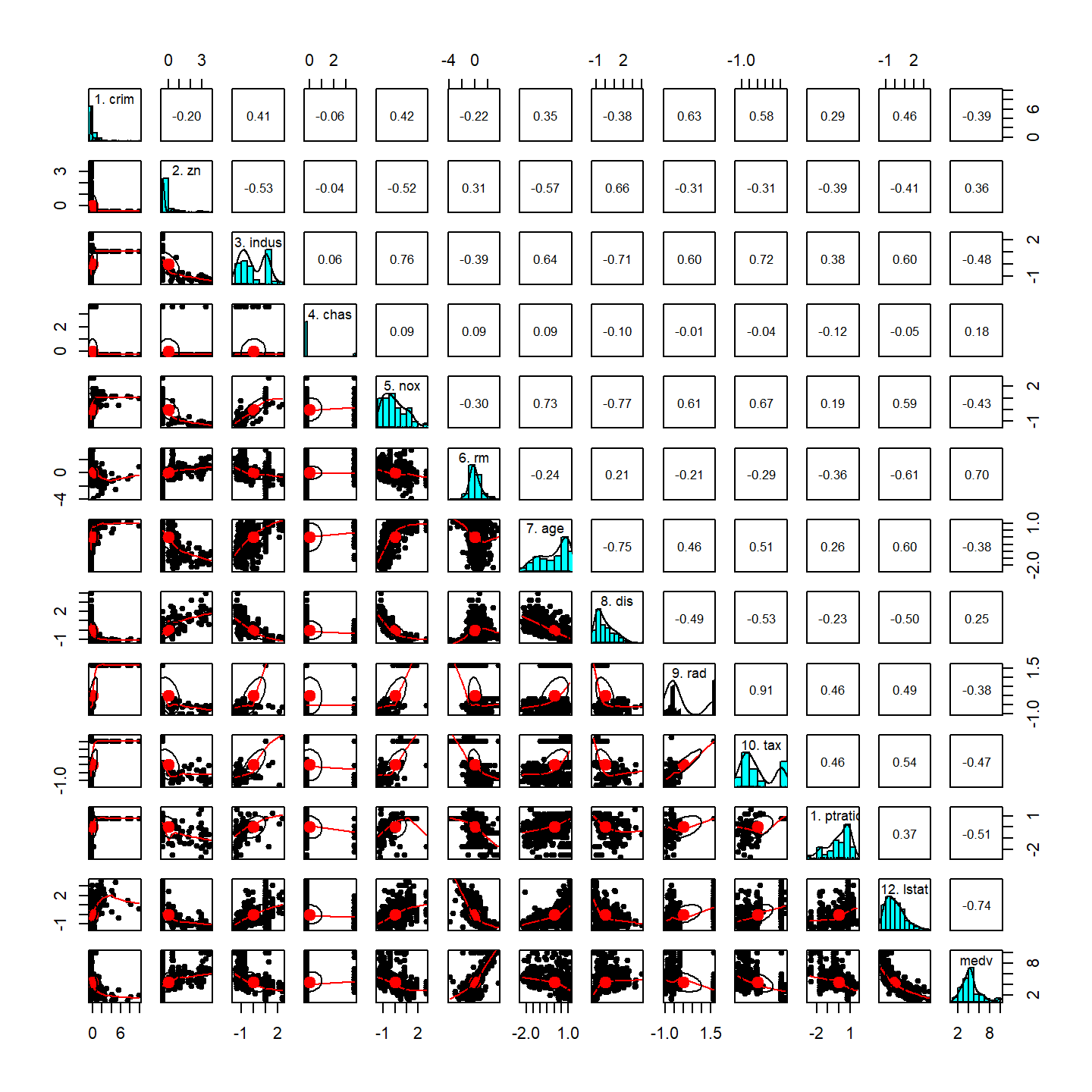}
\caption{Summary of the boston-housing dataset.}
\label{summary:boston-housing}
\end{figure}

\begin{figure}[!ht]
\centering 
\includegraphics[width=\textwidth]{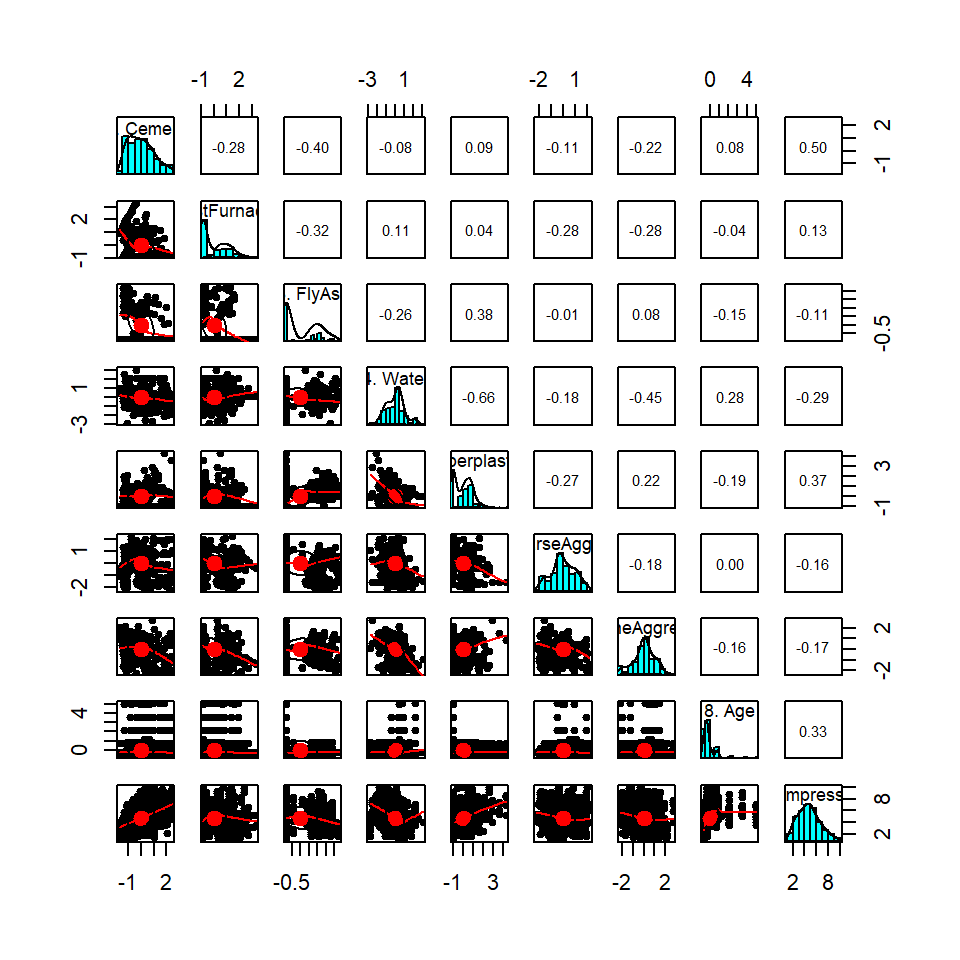}
\caption{Summary of the concrete dataset.}
\label{summary:concrete}
\end{figure}

\begin{figure}[!ht]
\centering 
\includegraphics[width=0.8\textwidth]{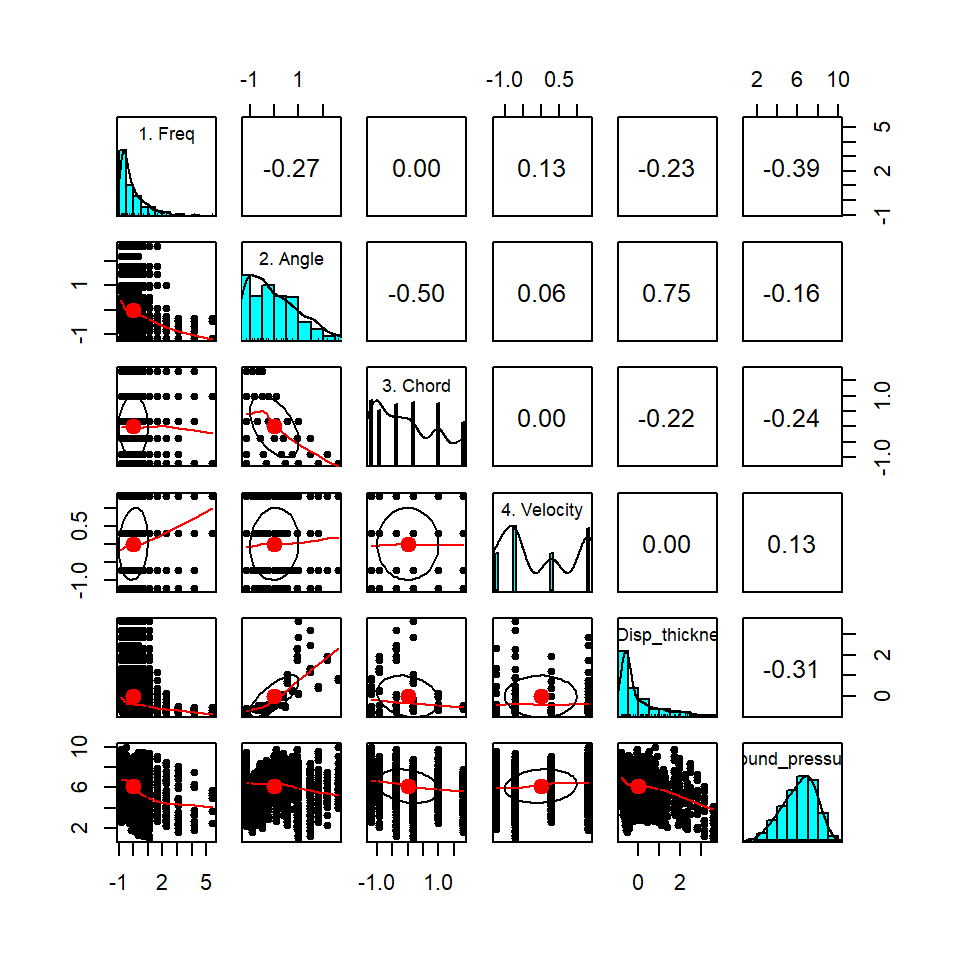}
\caption{Summary of the airfoil dataset.}
\label{summary:airfoil}
\end{figure}

\clearpage 

\section{Different schedule for decreasing learning rate}

The learning rate in MPS algorithm is multiplied by $0.95$ for each $50$ iteration, in both synthetic and real-world dataset experiments. 
To demonstrate the robustness against the different schedule for learning rate, we employ another schedule for the learning rate: the rate is multiplied by $0.97$ for each $100$ iteration (i.e., the decay is slower): see Figure~\ref{fig:real-estate-20000-100-097}.

In our experiments, the N$^3$POM, with its $270$ parameters ($L=50,R=20,d=3$), offers significant flexibility. This makes accurately determining the exact value of the coefficient function $\bs b_*(u)$, particularly in its tail region, challenging in a small dataset. Despite this, N$^3$POM demonstrates a consistent ability to capture the overall trend, as evidenced by comparing with Figure~\ref{fig:real-estate-main}(\subref{fig:real-estate}), even though slight fluctuation in its output is observed.

\begin{figure}[!ht]
\centering
\includegraphics[width=\textwidth]{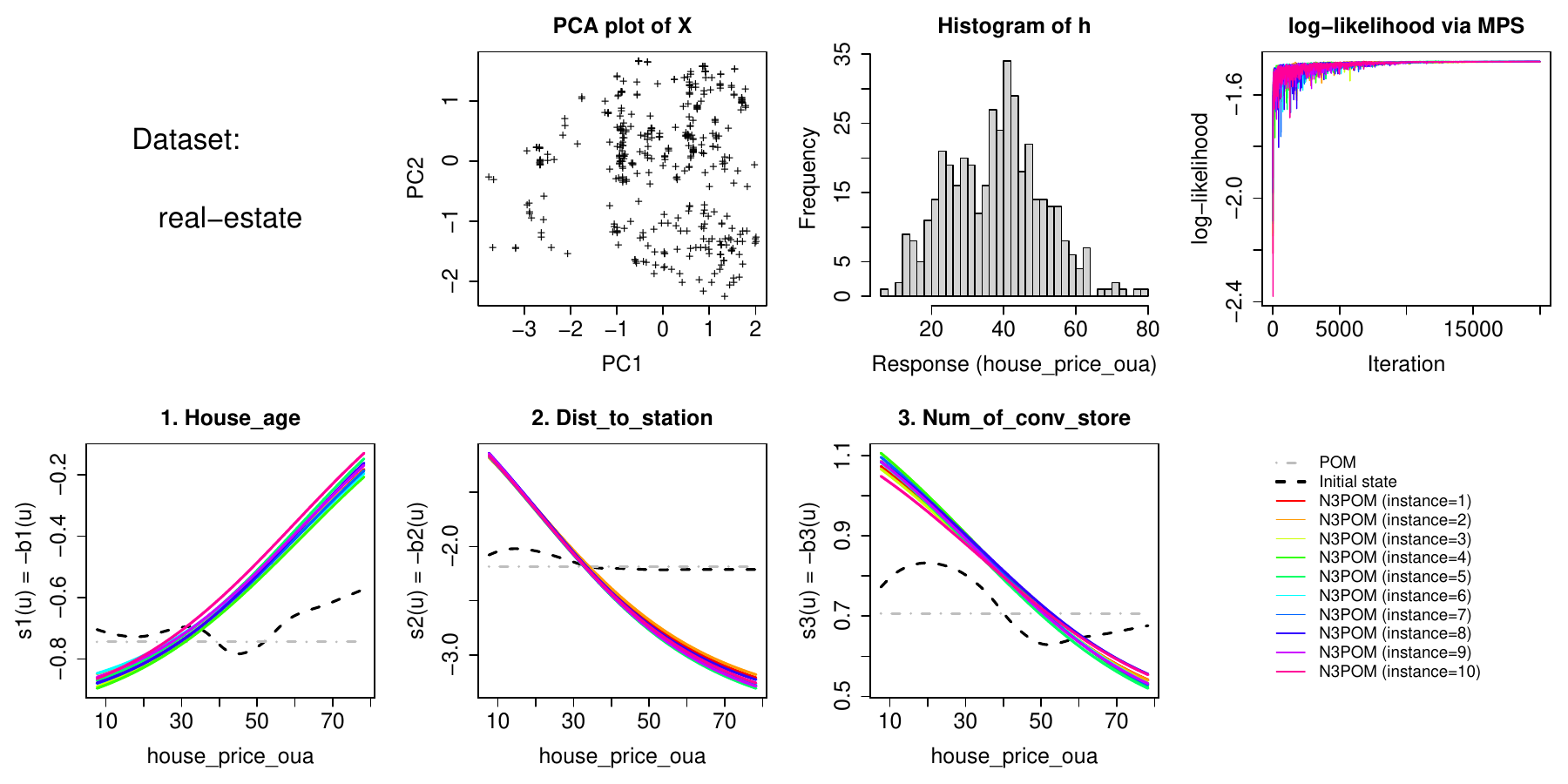}
\caption{Real estate dataset experiment.}
\label{fig:real-estate-20000-100-097}
\end{figure}

\clearpage
\section{Adaptation to the discrete responses}
\label{supp:additive_perturbation}

While we consider the response $H \in [1,J]$ taking a value in the connected interval $[1,J]$, it would be worthwhile to consider the adaptation to the discrete response $G \in \{1,2,\ldots,J\}$ because of the variety of applications. 
Although the proposed N$^3$POM~\eqref{eq:cmccp} can be formally trained using the discrete response, the discrete responses are not sufficient to fully train the continuous model $f_u(\bs x)$. 
See Remark~\ref{remark:log-likelihood_of_NPOM} for the log-likelihood for the interval-censored data, which is typically used to train NPOM, and its relation to the likelihood considered in this study. 
In our experiments with various discretized synthetic and real-world datasets, we found that N$^3$POM trained using either (i) the log-likelihood for interval-censored data or (ii) the log-likelihood~\eqref{eq:proposed_log_lik} considered in this study both exhibit similar behavior to the POM. This means that the estimated coefficients $\bs b(u)$ become constant, even when the underlying function depends on the response variable. 
To address this issue, we propose incorporating an additive perturbation to the responses so that the responses take values in not only $\{1,2,\ldots,J\}$ but also the connected interval $[1,J]$.

Our idea is simple. If we have the discrete responses $g_1,g_2,\ldots,g_n \in \{1,2,\ldots,J\}$, we can uniformly generate the random numbers $e_1,e_2,\ldots,e_n$ i.i.d. over the region $[-0.5,0.5]$ and add $e_1,e_2,\ldots,e_n$ to $g_1,g_2,\ldots,g_n$, respectively. Finally, we round the obtained responses to take the values between $1$ and $J$. 
Namely, we define a (random) perturbation operator
\begin{align}
\cut{g_i}
:=
\argmin_{j \in [1,J]} | (g_i + e_i) - j|,
\quad 
(i \in \{1,2,\ldots,n\}).
\label{eq:perturbation_operator}
\end{align}

Although heuristic, this perturbation operator~\eqref{eq:perturbation_operator} is explainable in the context of ordinary least squares regression. As is well-known in asymptotic theory, the estimated regression function in ordinary least squares converges to the conditional expectation $f_*(X)=\mathbb{E}[G \mid X]$; so adding the mean-zero perturbation, $E \sim U[-0.5,0.5]$ is expected not to cause any bias (i.e., $\mathbb{E}[G+E \mid X]=\mathbb{E}[G \mid X]=f_*(X)$). 
A similar result is expected to hold for ordinal regression. 
While the estimation efficiency slightly decreases because of this perturbation, it is advantageous to fully train the continuous NN; see the numerical experiments in Section~\ref{sec:experiments_synthetic} for the effectiveness of the additive perturbation~\eqref{eq:perturbation_operator}.

\section{Remarks on interpretation}
\label{supp:marginal_effects}

While the coefficients $\bs s(u)=-\bs b(u)$ considered in this study are expected to represent the influence of each covariate to the CCP, more strictly speaking, they in fact show the influence on the logit function applied to the CCP. 
Namely, the coefficients are also influenced by the logit function, and the coefficients in the tail region ($u \approx 0,u \approx J$) tend to be amplified; we may employ a marginal effect~\citep{agresti2018simple} $\frac{\partial}{\partial \bs x}\mathbb{P}_{\cnpom}(H > u \mid X=\bs x)=\bs s(u) \sigma^{[1]}(-f_u(\bs x))$ when considering the influence on the CCP directly. However, the marginal effect differs depending on the covariate $X$ and it tends to (excessively) shrink the influence to $0$ in the tail region $(u \approx 0,u \approx J)$ as $\sigma^{[1]}(-\infty)=\sigma^{[1]}(+\infty)=0$; unlike the simple coefficients $\bs s(u)$, marginal effect cannot capture the tendency whether the influence of the covariate increases or decreases (as $u$ increases), due to the shrinking behavior in the tail region. 
In the case of real-estate dataset above, the interesting coefficient $s_k(u)$ of house-age that approximates $0$ as $u \approx J$, cannot be detected when using the marginal effect, as almost all marginal effects approximates $0$ as $u \approx J$ (regardless of how the coefficient is important in the tail region). 
As a future work, it would be worthwhile to consider a more interpretable score to evaluate the influence of the covariates in the context of ordinal regression.

\end{document}